\documentclass[a4paper]{article}

\usepackage{natbib}
\usepackage[margin=1in]{geometry}
\usepackage{authblk}

\usepackage[utf8]{inputenc} 
\usepackage[T1]{fontenc}    

\usepackage{url}            
\usepackage{booktabs}       
\usepackage{amsfonts}       
\usepackage{nicefrac}       
\usepackage{microtype}      

\usepackage{makecell}

\usepackage{subcaption}

\usepackage{bm}
\usepackage{amsmath}
\usepackage{amsthm}
\newtheorem{lemma}{Lemma}
\newtheorem{remark}{Remark}
\usepackage{amssymb}
\usepackage{mathtools}
\usepackage{enumitem}
\usepackage{algorithm}
\usepackage{algpseudocode}

\usepackage{graphicx}
\usepackage{epstopdf}

\usepackage{hyperref}       

\usepackage{xr-hyper}

\title{Inference in Stochastic Epidemic Models via Multinomial Approximations}

\author{Nick Whiteley}
\author{Lorenzo Rimella}%
\affil{School of Mathematics, University of Bristol, and the Alan Turing Institute}

\begin{document}

\maketitle

\begin{abstract}
	We introduce a new method for inference in stochastic epidemic models which uses recursive multinomial approximations to integrate over unobserved variables and thus circumvent likelihood intractability. The method is applicable to a class of discrete-time, finite-population compartmental models with partial, randomly under-reported or missing count observations. In contrast to state-of-the-art alternatives such as Approximate Bayesian Computation techniques, no forward simulation of the model is required and there are no tuning parameters. Evaluating the approximate marginal likelihood of model parameters is achieved through a computationally simple filtering recursion. The accuracy of the approximation is demonstrated through analysis of real and simulated data using a model of the 1995 Ebola outbreak in the Democratic Republic of Congo. We show how the method can be embedded within a Sequential Monte Carlo approach to estimating the time-varying reproduction number of COVID-19 in Wuhan, China, recently published by \cite{kucharski2020early}.
\end{abstract}

\section{Introduction}\label{sec:intro}

Compartmental models are used for predicting the scale and duration of epidemics, estimating epidemiological parameters such as reproduction numbers, and guiding outbreak control measures  \citep{brauer2008compartmental, o2010introduction, kucharski2020early}. They are increasingly important because they allow joint modelling of disease dynamics and multimodal data, such as medical test results, cell phone and transport flow data \citep{rubrichi2018comparison, wu2020nowcasting}, census and demographic information \citep{prem2020effect}. However, statistical inference in stochastic variants of compartmental models is a major computational challenge \citep{breto2018modeling}. The likelihood function for model parameters is usually intractable because it involves summation over a prohibitively large number of configurations of latent variables representing counts of subpopulations in disease states which cannot be observed directly.




This has lead to the recent development of sophisticated computational methods for approximate inference involving various forms of stochastic simulation \citep{funk2020choices}. Examples include Approximate Bayesian Computation (ABC) \citep{kypraios2017tutorial,mckinley2018approximate,brownABC,brown2016empirically}, Data Augmentation Markov Chain Monte Carlo (MCMC) \citep{lekone2006statistical}, Particle Filters \citep{murray2018delayed},  Iterated Filtering \citep{stocks2019iterated}, and Synthetic Likelihood \citep{fasiolo2016comparison}. These methods continue to have real public health impact, for example the ABSEIR R package of \cite{brownABC} features in the current UK COVID-19 surveillance protocols  \citep{de2020emergence}. However the intricacy of these methods, their substantial computational cost arising from use of stochastic simulation, and their dependence on tuning parameters are obstacles to their wider use and scalability. The present work addresses the challenge of finding alternative inference techniques which are computationally simple and easy to use.



\paragraph{Contributions} We introduce a new approach to inference in  compartmental epidemic models which:
\begin{itemize}[leftmargin=*]
	\item applies to a class of finite population, partially observed, discrete-time, stochastic models. In contrast to ODE models, these models can account for statistical variability in disease dynamics;
	\item allows approximate evaluation of the likelihood function for model parameters, and filtering and smoothing for compartment occupation numbers, without any stochastic simulation or algorithm tuning parameters, in contrast to state-of-the-art techniques such as ABC; 
	\item revolves around a computationally simple filtering recursion. The resulting likelihood and smoothing approximations can be combined with e.g., MCMC or Expectation Maximization techniques for parameter estimation;
	\item  is shown to recover ground truth parameter values from synthetic data, and to compare favourably against Data Augmentation MCMC \citep{lekone2006statistical}, ABC using the ABSEIR R package \citep{brownABC} and ODE  \citep{chowell2004basic} alternatives analyzing real Ebola outbreak data under a model from \cite{lekone2006statistical};
	\item is used to extend a method of \cite{kucharski2020early} for estimating the time-varying reproduction number of COVID-19 in Wuhan, China, from an ODE compartmental model to a stochastic model.
\end{itemize}

\section{Preliminaries}

\subsection{Difficulties of inference in stochastic compartmental models}\label{subsec:determ_and_stoch_models}
We use the well-known Susceptible-Exposed-Infective-Recovered (SEIR) model as a simple running example. The new methods we propose are applied to more realistic and complex models in section \ref{sec:numerics}. 

\paragraph{SEIR example.} The discrete-time stochastic SEIR model is:
\begin{equation} \label{eq:SEIR_disc}
\begin{split}
& S_{t+1} = S_t - B_t,\\
& E_{t+1} = E_t + B_t - C_t,\\
& I_{t+1} = I_t +C_t - D_t , \\
& R_{t+1} = R_t + D_t,
\end{split}
\end{equation}
with conditionally independent, binomially-distributed random variables: 
\begin{equation} \label{eq:SEIR_BCD}
\begin{split}
& B_t \sim \mathrm{Bin} (S_t , 1-e^{-h\beta I_t/n})\\
& C_t \sim \mathrm{Bin} (E_t , 1-e^{-h\rho}),\\ 
& D_t \sim \mathrm{Bin} (I_t , 1-e^{-h\gamma}), 
\end{split}
\end{equation}
where $h>0$ is a time-step size, $\beta,\rho,\gamma$ are model parameters, and the process is initialized with nonnegative integers in each of the compartments $(S_0, E_0, I_0, R_0)$ such that $S_0+E_0+I_0+R_0=n$ and $n$ is the total population size.  The interpretation of $\beta$ is the rate at which an interaction between a susceptible individual and the infective proportion of the population results in the disease being passed to that individual. The mean exposure and infective periods are respectively $1/\rho$ and $1/\gamma$.   The sequence $(S_t,E_t,I_t,R_t)_{t\geq0}$ is a Markov chain.



In practice, one typically observes times series of count data associated with some subset of the compartments, perhaps subject to random error, or under-reporting. Given such data, evaluating the likelihood function of the model parameters and initial condition requires the variables associated with unobserved compartments to be marginalized out. In general this is infeasible for models with anything but a small population size $n$ and a small number of compartments.

Stochastic compartmental models also commonly arise in the form of continuous-time pure jump Markov processes, in which transitions of individuals between compartments occur in an asynchronous manner \citep{breto2018modeling}. Likelihood-based inference for such processes is similarly intractable in general. There are rigorous limit theorems which link continuous time Markov process compartmental models to deterministic ODE models in the large population limit, e.g., \cite{kurtz1970solutions, kurtz1971limit, roberts2015nine}. However the precise nature of the asymptotic is somewhat subtle and not always meaningful in practice: the supplementary materials includes a simple example in which a stochastic model exhibits substantial statistical variation even when the population size is $10^7$, and the corresponding ODE limit is pathological. Thus, ODE models are no substitute for stochastic models.


\subsection{Notation}\label{subsec:notation} In the remainder of the paper, bold upper-case and bold lower-case characters are respectively matrices and column vectors, e.g., $\mathbf{A}$ and $\mathbf{b}$, with generic elements $a^{(i,j)}$ and $b^{(i)}$. The length-$m$ column vector of 1's is denoted $\mathbf{1}_m$. A vector is called a probability vector if its elements are nonnegative and sum to $1$. A matrix is called row-stochastic if its elements are nonnegative and its row sums are all $1$. The indicator function is denoted $\mathbb{I}[\cdot]$. The element-wise product of matrices is denoted $\mathbf{A}\circ\mathbf{B}$ and the outer product of vectors is denoted $\mathbf{a}\otimes \mathbf{b}$. Element-wise natural logarithm and factorial are denoted $\log \mathbf{A} $ and $\mathbf{A}!$.  For positive integers $m$ and $n$, define $\mathcal{C}_m\coloneqq\{1,\ldots,m\}$ and $\mathcal{S}_{m,n} \coloneqq \{ \mathbf{x}=[x^{(1)}\,\cdots \,x^{(m)}]^{\mathrm{T}}:x^{(i)}\geq0, i=1,\ldots, m; \sum_{i=1}^m x^{(i)}=n\}$. For $\mathbf{x}\in\mathcal{S}_{m,n}$, define $\bm{\eta}(\mathbf{x})\coloneqq [x^{(1)}/n\,\cdots \,x^{(m)}/n]^{\mathrm{T}}$. For a matrix $\mathbf{P}$ (resp. a vector $\bm{\pi}$) with nonnegative elements summing to $1$, $\mathrm{Mult}(n,\mathbf{P})$ (resp. $\mathrm{Mult}(n,\bm{\pi})$) denotes the distribution of the random matrix (resp. vector)  whose elements are the incidence counts obtained from sampling $n$ times with replacement according to  $\mathbf{P}$  (resp. $\bm{\pi}$). This is the usual definition of a multinomial distribution.

\section{Model}\label{sec:model}

\subsection{A class of compartmental models specified by the transition probabilities of individuals}\label{sec:individ} 
The general model we consider is defined by:  $m$, the number of compartments; $n$, the total population size; a length-$m$ probability vector $\bm{\pi}_0$; and for each $t\geq1$, a mapping from length-$m$ probability vectors to $m\times m$ row-stochastic matrices, $\bm{\eta}\mapsto\mathbf{K}_{t,\bm{\eta}}$.  The population at time $t\geq0$ is a set of $n$ random variables $\{\xi_t^{(1)},\ldots,\xi_t^{(n)}\}$, each valued in $\mathcal{C}_{m}$. The counts of individuals in each of the $m$ compartments at time $t$ are collected in a vector $\mathbf{x}_t = [x_t^{(1)}\cdots x_t^{(m)}]^{\mathrm{T}}\in \mathcal{S}_{m,n}$, $x_t^{(i)} \coloneqq \sum_{j=1}^{n}\mathbb{I}[\xi_t^{(j)}=i]$. For $t\geq 1$ let $\mathbf{Z}_t$ be the $m\times m$ matrix  with elements $z_t^{(i,j)}\coloneqq\sum_{k=1}^n \mathbb{I}[\xi_{t-1}^{(k)}=i, \xi^{(k)}_{t}=j]$, which counts the individuals transitioning from compartment $i$ at $t-1$ to $j$ at $t$.

The sequence $\{\xi_t^{(1)},\ldots,\xi_t^{(n)}\}_{t\geq 0}$ is constructed to be a Markov chain: the  members of the initial population $\{\xi_0^{(1)},\ldots,\xi_0^{(n)}\}$ are i.i.d. with $p(\xi_0^{(i)}=j)=\bm{\pi}_{0}^{(j)}$,   and given $\{\xi_{t-1}^{(1)},\ldots,\xi_{t-1}^{(n)}\}$, $\{\xi_t^{(1)},\ldots,\xi_t^{(n)}\}$ are conditionally independent, with $\xi_{t}^{(i)}$ drawn from the $\xi_{t-1}^{(i)}$'th row of  $\mathbf{K}_{t,\bm{\eta}(\mathbf{x}_{t-1})}$.  It follows from this prescription that the sequence of matrices $(\mathbf{Z}_t)_{t\geq0}$ is also a Markov chain.  Indeed, conditional on $\mathbf{Z}_{t-1}$, and hence automatically on $\mathbf{x}_{t-1}$ since $ \mathbf{Z}_t\mathbf{1}_m = \mathbf{x}_{t-1}$, the rows of $\mathbf{Z}_t$ are independent, and the distribution of the $i$th row of $\mathbf{Z}_t$ is $\mathrm{Mult}(x_{t-1}^{(i)},\mathbf{K}_{t,\bm{\eta}(\mathbf{x}_{t-1})}^{(i,\cdot)})$, where $\mathbf{K}_{t,\bm{\eta}(\mathbf{x}_{t-1})}^{(i,\cdot)}$ is the $i$th row of $\mathbf{K}_{t,\bm{\eta}(\mathbf{x}_{t-1})}$.  Moreover, noting $\mathbf{1}_m^{\mathrm{T}} \mathbf{Z}_t = \mathbf{x}_{t}^{\mathrm{T}}$, we observe $(\mathbf{x}_t)_{t\geq0}$ is also a Markov chain, but we shall not need an explicit formula for its transition probabilities.


\paragraph{SEIR example} The SEIR model in (\ref{eq:SEIR_disc})-(\ref{eq:SEIR_BCD}) is equivalent to taking $m=4$,
\begin{equation}\label{eq:SEIR_as_ind}
\left (\mathbf{K}_{t,\bm{\eta}}\right )^{(i,j)}=
\begin{cases}
e^{-h\beta \eta^{(3)}}   & i=j=1 \\
1-e^{-h\beta \eta^{(3)}} & i=1 \text{ and } j=2 \\
e^{-h\rho}               & i=j=2 \\
1-e^{-h\rho} & i=2 \text{ and } j=3 \\
e^{-h\gamma}             & i=j=3 \\
1-e^{-h\gamma}           & i=3 \text{ and } j=4 \\
1                        & i =j= 4 \\
0                        & \text{otherwise}
\end{cases}
\end{equation}
for all $t\geq 1$, and identifying $ [x_t^{(1)}\,x_t^{(2)}\,x_t^{(3)}\, x_t^{(4)}]^{\mathrm{T}}$ with  respectively the counts of susceptible, exposed, infective and recovered individuals at time $t$.

We consider two observation models.


\subsection{Observations derived from $(\mathbf{x}_t)_{t\geq1}$} \label{sec:obs_model_x}
In this scenario, the observation at time $t\geq1$ is a length-$m$ vector $\mathbf{y}_t $ with elements $y_t^{(i)}$ which are conditionally independent given $\mathbf{x}_t$, and:
\begin{equation}\label{eq:obs_model_comp}
y_t^{(i)}\sim\mathrm{Bin}(x_t^{(i)},q_t^{(i)}).
\end{equation}
We shall collect the parameters $q_t^{(i)}\in[0,1]$ in a length-$m$ vector $\mathbf{q}_t$. When conducting likelihood-based inference for $\mathbf{x}_t$ using this model, if $y_t^{(i)}$ is a missing observation, then in the likelihood function associated with  (\ref{eq:obs_model_comp}) one should take $y_t^{(i)}$ to be $0$,  set $q_t^{(i)}=0$.

\subsection{Observations derived from $(\mathbf{Z}_t)_{t\geq1}$} \label{sec:obs_model_z}
In this scenario, the observation at time $t\geq 1$ is a $m\times m$ matrix $\mathbf{Y}_t $ with elements $y_t^{(i,j)}$ which are conditionally independent given $\mathbf{Z}_t$, and:
\begin{equation}\label{eq:obs_model_inc}
y_t^{(i,j)}\sim\mathrm{Bin}(z_t^{(i,j)},q_t^{(i,j)}).
\end{equation}
The  parameters $q_t^{(i,j)}\in[0,1]$ from (\ref{eq:obs_model_inc}) are collected into a $m\times m$ matrix $\mathbf{Q}_t$. Missing data are handled by putting a $0$ in place of the missing $y_t^{(i,j)}$ and setting $q_t^{(i,j)}=0$.

\paragraph{SEIR example} In practice, one typically observes, at each time step, counts of \emph{new} infectives rather than the total number of infectives, subject to some random under-reporting or missing data. How can such data be represented in the model? Due to the definition of $\mathbf{Z}_t$, the number of new infectives at time $t$ is exactly $z_t^{(2,3)}$, since the only way an individual can transition to being infective (compartment 3) is by first being exposed (compartment 2). Therefore in this case $y_t^{(2,3)}$ following (\ref{eq:obs_model_inc}) is a count of newly infectives at time $t$, subject to binomial random under-reporting with parameter $q_t^{(2,3)}$, as required.

\section{Inference}\label{sec:inference}
We now introduce our methods for approximating the so-called filtering distributions and marginal likelihoods $p(\mathbf{x}_t| \mathbf{y}_{1:t})$,  $p(\mathbf{y}_{1:t})$  and $p(\mathbf{Z}_t| \mathbf{Y}_{1:t})$, $p(\mathbf{Y}_{1:t})$ under respectively the observation models of sections \ref{sec:obs_model_x} and \ref{sec:obs_model_z}. These quantities are at the core of smoothing and parameter estimation techniques demonstrated in section \ref{sec:numerics} and detailed in the supplementary materials.

For the observation model of section \ref{sec:obs_model_x}, note that $(\mathbf{x}_t,\mathbf{y}_t)_{t\geq1}$ is a hidden Markov model, and  in principle the filtering distributions can be computed through a two-step recursion:
\begin{equation}\label{eq:exact_recursion}
p(\mathbf{x}_{t-1}|\mathbf{y}_{1:t-1})  \stackrel{\mathrm{prediction}}{\longrightarrow} p(\mathbf{x}_t |\mathbf{y}_{1:t-1})\stackrel{\mathrm{update}}{ \longrightarrow} p(\mathbf{x}_t |\mathbf{y}_{1:t}). 
\end{equation}
However in practice, the summations involved in the `prediction' and `update' operations are prohibitively expensive since they involve summing over all possible values of $\mathbf{x}_{t-1}$ and $\mathbf{x}_{t}$.
\subsection{Approximating the prediction operation}\label{sec:approx_prediction}
For each $\mathbf{x}=[x^{(1)}\,\cdots\, x^{(m)}]^{\mathrm{T}}\in\mathcal{S}_{m,n}$ and length-$m$ probability vector $\bm{\eta}$, let $M_t(\mathbf{x},\bm{\eta},\cdot)$ be the probability mass function on $\mathcal{S}_{m,n}$ of $(\mathbf{1}_m^{\mathrm{T}} \mathbf{Z})^{\mathrm{T}}$, where $\mathbf{Z}$ is a  random $m\times m$ matrix whose rows are independent, and whose $i$th row has distribution $\mathrm{Mult}(x^{(i)},\mathbf{K}_{t,\bm{\eta}}^{(i,\cdot)})$. So by construction $M_t(\mathbf{x}_{t-1},\bm{\eta}(\mathbf{x}_{t-1}),\mathbf{x}_t)$ is the probability transition function for the Markov chain $(\mathbf{x}_t)_{t\geq0}$ defined in section \ref{sec:individ}. Thus the prediction operation in (\ref{eq:exact_recursion}) can be written in terms of $M_t$: 
\begin{equation}
\begin{split}
&p(\mathbf{x}_{t}|\mathbf{y}_{1:t-1})= \sum_{\mathbf{x} _{t-1}\in \mathcal{S}_{m,n}} p(\mathbf{x}_{t-1}|\mathbf{y}_{1:t-1}) p(\mathbf{x}_{t}|\mathbf{x}_{t-1}) \\
&=\sum_{\mathbf{x} _{t-1}\in \mathcal{S}_{m,n}} p(\mathbf{x}_{t-1}|\mathbf{y}_{1:t-1}) M_t(\mathbf{x}_{t-1},\bm{\eta}(\mathbf{x}_{t-1}),\mathbf{x}_t) .\label{eq:prediction}
\end{split}
\end{equation}
Our approximation to this operation is as follows: assuming we have already obtained a multinomial distribution approximation to $p(\mathbf{x}_{t-1}|\mathbf{y}_{1:t-1})$, then in (\ref{eq:prediction}) we replace $p(\mathbf{x}_{t-1}|\mathbf{y}_{1:t-1})$ by this multinomial distribution, and replace the vector $\bm{\eta}(\mathbf{x}_{t-1})$ by its expectation under this multinomial distribution. This results in a multinomial distribution approximation to $p(\mathbf{x}_{t}|\mathbf{y}_{1:t-1})$. The following lemma formalizes this recipe.
\begin{lemma}\label{lem:prediction}
	If for a given length-$m$ probability vector $\bm{\pi}$,  $\mu(\cdot)$ is the probability mass function on $\mathcal{S}_{m,n}$ associated with $\mathrm{Mult}(n,\bm{\pi})$ and $\mathbb{E}_{\mu}[\bm{\eta}(\mathbf{x})]$ is the expected value of $\bm{\eta}(\mathbf{x})$ when $\mathbf{x}\sim \mu$, then $\sum_{\mathbf{x} \in \mathcal{S}_{m,n}} \mu(\mathbf{x}) M_t(\mathbf{x},\mathbb{E}_{\mu}[\bm{\eta}(\mathbf{x})],\cdot)$ is the probability mass function associated with $\mathrm{Mult}(n,\bm{\pi}^{\mathrm{T}}\mathbf{K}_{t,\bm{\pi}})$.
\end{lemma}
The proof is given in the supplementary materials.

\subsection{Approximating the update operation}\label{sec:approx_update}
The update operation in (\ref{eq:exact_recursion}) is:
\begin{equation}\label{eq:update}
\begin{split}
&p(\mathbf{x}_t |\mathbf{y}_{1:t}) = \frac{p(\mathbf{y}_t |\mathbf{x}_t)p(\mathbf{x}_t |\mathbf{y}_{1:t-1})} {p(\mathbf{y}_t|\mathbf{y}_{1:t-1})} ,\\ &p(\mathbf{y}_t|\mathbf{y}_{1:t-1}) = \sum_{\mathbf{x}_t\in\mathcal{S}_{m,n}}p(\mathbf{y}_t |\mathbf{x}_t)p(\mathbf{x}_t |\mathbf{y}_{1:t-1}),
\end{split}
\end{equation}
which has the interpretation of a Bayes' rule update applied to $p(\mathbf{x}_t |\mathbf{y}_{1:t-1})$. Assuming we have already obtained a multinomial distribution approximation to $p(\mathbf{x}_t |\mathbf{y}_{1:t-1})$, our approximation to the update operation is to substitute this multinomial distribution  in place of $p(\mathbf{x}_t |\mathbf{y}_{1:t-1})$ in (\ref{eq:update}), resulting in a shifted-multinomial distribution whose mean vector is used to define a multinomial distribution approximation to $p(\mathbf{x}_t |\mathbf{y}_{1:t})$. The following lemma formalizes this recipe.
\begin{lemma}\label{lem:update}
	Suppose that  $\mathbf{x}\sim\mathrm{Mult}(n,\bm{\pi})$ for a given length-$m$ probability vector $\bm{\pi}$, and assume that given $\mathbf{x}$, $\mathbf{y}$ is a vector with conditionally independent elements distributed: $y^{(i)}\sim\mathrm{Bin}(x^{(i)},q^{(i)})$. Then the conditional distribution of $\mathbf{x}$ given $\mathbf{y}$ is equal to that of $\mathbf{y} + \mathbf{x}^\star$,   where
	\begin{equation}
	\mathbf{x}^\star\sim \mathrm{Mult}\left(n-\mathbf{1}_m^{\mathrm{T}} \mathbf{y},  \dfrac{\bm{\pi} \circ (\mathbf{1}_m-\mathbf{q})}{1 - \bm{\pi}^{\mathrm{T}} \mathbf{q}}\right)\label{eq:x_star_dist}
	\end{equation}
	with $\mathbf{q}=[q^{(1)}\,\cdots\,q^{(m)}]^\mathrm{T}$, and the conditional mean of  $\mathbf{x}$  given  $\mathbf{y}$ is:
	\begin{equation}
	\mathbb{E}[\mathbf{x}|\mathbf{y}]= y + (n-\mathbf{1}_m^{\mathrm{T}} \mathbf{y})\left(\dfrac{\bm{\pi} \circ (\mathbf{1}_m-\mathbf{q})}{1 - \bm{\pi}^{\mathrm{T}} \mathbf{q}}\right).\label{eq:cond_mean}
	\end{equation}
	Moreover, the marginal distribution of $\mathbf{y}$ has probability mass function given by:
	\begin{equation}
	\begin{split}
	\log p(\mathbf{y}) &= \log(n!) + \mathbf{y}^{\mathrm{T}} (\log \bm{\pi}+\log \mathbf{q} ) - \mathbf{1}_m^{\mathrm{T}} \log(\mathbf{y}!)  \\
	&+ (n-\mathbf{1}_m^{\mathrm{T}}\mathbf{y})\log (1-\bm{\pi} ^\mathrm{T}\mathbf{q}) - \log( (n-\mathbf{1}_m^{\mathrm{T}}\mathbf{y})! ),\label{eq:marg_y}
	\end{split}
	\end{equation}
	with the convention $0 \log 0 \equiv 0$.
\end{lemma}
The proof is given in the supplementary materials.
\subsection{Multinomial filtering}
Putting together the results of lemma \ref{lem:prediction} and lemma \ref{lem:update} in a recursive fashion leads us to algorithm \ref{alg:filtering_comp}; line 3 is motivated by lemma \ref{lem:prediction}, line 4 is motivated by (\ref{eq:x_star_dist})-(\ref{eq:cond_mean}).

\begin{algorithm}[h]
	\caption{Multinomial filtering with observations derived from $(\mathbf{x}_t)_{t\geq1}$}\label{alg:filtering_comp}
	\begin{algorithmic}[1]
		\State \textbf{initialize} $\bm{\pi}_{0|0}\leftarrow\bm{\pi}_0$
		\For{$t\geq 1$}
		\State $\bm{\pi}_{t|t-1}\leftarrow (\bm{\pi}_{t-1|t-1}^{\mathrm{T}} \mathbf{K}_{t,\bm{\pi}_{t-1|t-1}} )^\mathrm{T}$ %
		
		\State $\bm{\pi}_{t|t}\leftarrow \dfrac{\mathbf{y}_{t}}{n}+\left( 1-\dfrac{\mathbf{1}_m^{\mathrm{T}} \mathbf{y}_{t} }{n}\right)
		\dfrac{\bm{\pi}_{t|t-1} \circ (\mathbf{1}_m-\mathbf{q}_{t})}{1 - \bm{\pi}_{t|t-1}^{\mathrm{T}} \mathbf{q}_{t}}$
		\State{$\log w_t \leftarrow \log(n!) + \mathbf{y}_t^{\mathrm{T}} (\log \bm{\pi}_{t|t-1}+\log \mathbf{q}_t )- \mathbf{1}_m^{\mathrm{T}} \log(\mathbf{y_t}!) + (n-\mathbf{1}_m^{\mathrm{T}}\mathbf{y}_t)\log (1-\bm{\pi}_{t|t-1} ^\mathrm{T}\mathbf{q}_t)- \log( (n-\mathbf{1}_m^{\mathrm{T}}\mathbf{y}_t)! )$}

		\EndFor
	\end{algorithmic}
\end{algorithm}

One may take as output from algorithm \ref{alg:filtering_comp} the approximations:
\begin{equation}\label{eq:alg1_approx_filt}
\begin{split}
&p(\mathbf{x}_t|\mathbf{y}_{1:t-1})  \approx \mathrm{Mult}(n,\bm{\pi}_{t|t-1}),\\ &p(\mathbf{x}_t|\mathbf{y}_{1:t}) \stackrel{d}{\approx} \mathbf{y}_{t} + \mathbf{x}_t^\star,
\end{split}
\end{equation}
where the $\stackrel{d}{\approx} $ term indicates approximation of $p(\mathbf{x}_t|\mathbf{y}_{1:t}) $ by the distribution of the sum of $\mathbf{y}_t$ (regarded as a constant) and a random variable $\mathbf{x}_t^\star$ which is defined to have distribution:  
\begin{equation}\label{eq:x_star}
\mathbf{x}_t^\star\sim \mathrm{Mult}\left(n-\mathbf{1}_m^{\mathrm{T}} \mathbf{y}_{t} ,  \dfrac{\bm{\pi}_{t|t-1} \circ (\mathbf{1}_m-\mathbf{q}_{t})}{1 - \bm{\pi}_{t|t-1}^{\mathrm{T}} \mathbf{q}_{t}}\right).
\end{equation}
In view of (\ref{eq:marg_y}), the quantities $w_t$ computed in algorithm \ref{alg:filtering_comp} can be used to approximate the marginal likelihood as follows:
\begin{equation}\label{eq:marg_like_approx}
p(\mathbf{y}_{1:t}) = p(\mathbf{y}_{1})\prod_{s=2}^{t} p(\mathbf{y}_{s}|\mathbf{y}_{1:s-1})\approx \prod_{s=1}^t w_s.
\end{equation}

Now turning to the observation model from section \ref{sec:obs_model_z} and noting that  $(\mathbf{Z}_t,\mathbf{Y}_t)_{t\geq1}$ is a hidden Markov model, we approximate the recursion:
\begin{equation}\label{eq:exact_recursion_Z}
p(\mathbf{Z}_{t-1}|\mathbf{Y}_{1:t-1})  \stackrel{\mathrm{prediction}}{\longrightarrow} p(\mathbf{Z}_t |\mathbf{Y}_{1:t-1})\stackrel{\mathrm{update}}{ \longrightarrow} p(\mathbf{Z}_t |\mathbf{Y}_{1:t}).  
\end{equation}
Many details are similar to those above so are given in the supplementary materials. The counterpart of algorithm \ref{alg:filtering_comp} is algorithm \ref{alg:filtering_inc}, from which one may take the approximations: 
\begin{equation}\label{eq:alg2_approx_filt}
\begin{split}
&p(\mathbf{Z}_t|\mathbf{Y}_{1:t-1})  \approx \mathrm{Mult}(n,\mathbf{P}_{t|t-1}),\\ &p(\mathbf{Z}_t|\mathbf{Y}_{1:t}) \stackrel{d}{\approx} \mathbf{Y}_{t} + \mathbf{Z}_t^\star,
\end{split}
\end{equation}
where
\begin{equation}
\mathbf{Z}_t^\star\sim \mathrm{Mult}\left(n-\mathbf{1}_m^{\mathrm{T}} \mathbf{Y}_t \mathbf{1}_m , \dfrac{\mathbf{P}_{t|t-1}\circ(\mathbf{1}_m\otimes \mathbf{1}_m-\mathbf{Q}_{t})}{1 - \mathbf{1}_m^{\mathrm{T}}(\mathbf{P}_{t|t-1}\circ \mathbf{Q}_{t})\mathbf{1}_m}\right).
\end{equation}
The marginal likelihood is approximated using the same formula as in (\ref{eq:marg_like_approx}) but with the $w_t$'s computed as per algorithm \ref{alg:filtering_inc}.
\begin{algorithm}
	\caption{Multinomial filtering with observations derived from $(\mathbf{Z}_t)_{t\geq1}$}\label{alg:filtering_inc}
	\begin{algorithmic}[1]
		\State \textbf{initialize} $\bm{\pi}_{0|0}\leftarrow\bm{\pi}_0$
		\For{$t\geq 1$}
		\State $\mathbf{P}_{t|t-1} \leftarrow (\bm{\pi}_{t-1|t-1} \otimes \mathbf{1}_m ) \circ \mathbf{K}_{t,\bm{\pi}_{t-1|t-1}}$
		\State $\mathbf{P}_{t|t}\leftarrow\dfrac{\mathbf{Y}_{t}}{n}+\dfrac{\mathbf{P}_{t|t-1}\circ(\mathbf{1}_m\otimes \mathbf{1}_m-\mathbf{Q}_{t})}{1 - \mathbf{1}_m^{\mathrm{T}}(\mathbf{P}_{t|t-1}\circ \mathbf{Q}_{t})\mathbf{1}_m}-
		\left(\dfrac{\mathbf{1}_m^{\mathrm{T}}\mathbf{Y}_{t}\mathbf{1}_m}{n}\right)\dfrac{\mathbf{P}_{t|t-1}\circ(\mathbf{1}_m\otimes \mathbf{1}_m-\mathbf{Q}_{t})}{1 - \mathbf{1}_m^{\mathrm{T}}(\mathbf{P}_{t|t-1}\circ \mathbf{Q}_{t})\mathbf{1}_m}$
		\State $\log w_t \leftarrow \log(n!) + \mathbf{1}_m^{\mathrm{T}} (\mathbf{Y}_t \circ \log \mathbf{P}_{t|t-1} )\mathbf{1}_m + \mathbf{1}_m^{\mathrm{T}} (\mathbf{Y}_t \circ \log \mathbf{Q}_{t} )\mathbf{1}_m -  \mathbf{1}_m^{\mathrm{T}} \log(\mathbf{Y}_t !)\mathbf{1}_m+ (n-\mathbf{1}_m^{\mathrm{T}} \mathbf{Y}_t \mathbf{1}_m) \log (1 -\mathbf{1}_m^{\mathrm{T}} (\mathbf{P}_{t|t-1} \circ \mathbf{Q}_t) \mathbf{1}_m ) - \log ((n-\mathbf{1}_m^{\mathrm{T}} \mathbf{Y}_t \mathbf{1}_m)!)$
		\State $\bm{\pi}_{t|t}\leftarrow(\mathbf{1}_m^{\mathrm{T}}\mathbf{P}_{t|t})^{\mathrm{T}}$
		
		\EndFor
	\end{algorithmic}
\end{algorithm}

%

\subsection{Computational cost}

The computational cost of algorithms \ref{alg:filtering_comp} and \ref{alg:filtering_inc} is independent of the overall population size $n$,  except through factorial terms such as $\log(n!)$ and $\log((n-\mathbf{1}_m^\mathrm{T}\mathbf{y})!)$. However these terms do not depend on the model parameters $\mathbf{K}_{t,\bm{\eta}}$, $\mathbf{q}_t$ etc., so can be pre-computed or even not computed at all if the approximate marginal likelihood needs to be evaluated only up to a constant of proportionality independent of model parameters. Leaving these factorial terms out the worst case costs of algorithms \ref{alg:filtering_comp} and \ref{alg:filtering_inc} are therefore respectively $\mathcal{O}(t m^2)$ and $\mathcal{O}(t m^3)$.  Costs may be substantially lower in practice as $\mathbf{K}_{t,\bm{\eta}}$ and $\mathbf{q}_t$ are typically sparse. Similar observations hold for the smoothing algorithms.

This compares to $O(t m f(n))$ to simulate $(\mathbf{x}_t)_{t\geq0}$ from the model where $f(n)$ is the complexity of sampling from $\mathrm{Bin}(n,p)$, assuming no more than two non-zero entries in each row of $\mathbf{K}_{t,\bm{\eta}}$. A larger number of non-zero entries would imply a higher cost. Such a simulation is necessary (but usually not sufficient) to approximately evaluate the likelihood in ABSEIR \citep{brownABC}. The worst case is $f(n)=O(n)$, but modest improvements are available if one accepts `with high probability' performance measures \citep{farach2015exact}. The worst case time complexity of the Data Augmentation MCMC method \citep{lekone2006statistical} is also linear in $n$. Whilst the wall-clock time of any given algorithm is of course heavily dependent on exactly how it is implemented, these considerations suggest that the proposed methods will have attractive computational costs in many applications, where $m$ is often many orders of magnitude smaller than $n$

\section{Numerical results}\label{sec:numerics}

Additional details of models, data sources, algorithms, prior distributions, hyper-parameter settings, further numerical results and tutorials are given in the supplementary materials. 

\subsection{The 1995 Ebola outbreak in the Democratic Republic of Congo}\label{sec:ebola}
We analyzed simulated and real data under a discrete-time SEIR model used by \cite{lekone2006statistical} to investigate the impact of control interventions on the 1995 outbreak of Ebola in the Democratic Republic of Congo. Our experiments follow closely those in \cite{lekone2006statistical} to allow comparisons with their Data Augmentation MCMC method. We also include comparisons to the ABC method from the ABSEIR R package \citep{brownABC}, and results of least-squares fitting of an ODE model from \cite{chowell2004basic} which \cite{lekone2006statistical} used as a benchmark.

The model of \cite{lekone2006statistical} is the same as the SEIR model in (\ref{eq:SEIR_disc}) with $h=1$, except that $\beta$ is replaced by a time-varying parameter
$\beta_t = \beta$ for $t<t_\star$ and $\beta_t = \beta e^{-\lambda(t-t_\star)}$ for $t
\geq t_\star$ where $t_\star$ is the time at which control measures began. Thus $\mathbf{K}_{t,\bm{\eta}}$ is as in (\ref{eq:SEIR_as_ind}) but with $\beta$ replaced by this $\beta_t$. Also following \cite{lekone2006statistical},  the data consist of daily counts of new cases (i.e. new infectives) and new deaths (i.e. new removals). In \cite{lekone2006statistical} it was assumed these counts are observed directly, subject to known proportions of missing data. We consider a slightly more general observation model as per section \ref{sec:obs_model_z} with $q_t^{(i,j)}=0$ for all $(i,j)$ except $(2,3)$ and $(3,4)$, and where $q_t^{(2,3)}$ and $q_t^{(3,4)}$ are treated as constant-in-$t$ but otherwise unknown and to be estimated. 

\paragraph{Synthetic data}
Using the following settings from \cite{lekone2006statistical}:  $(\beta,\lambda,\rho,\gamma)=(0.2,0.2,0.2,0.143)$, $S_0=5,364,500$, $E_0=1$, $I_0=R_0=0$, $t_{\star}=130$ , plus $q_t^{(2,3)}=291/316$ and $q_t^{(3,4)}=236/316$ for all $t\geq1$ informed by realistic proportions of non-missing data \citep{lekone2006statistical}, we simulated the epidemic from the model until extinction, which took $175$ time steps. Table \ref{tab:ebola_synthetic} shows MLE's from an EM algorithm which uses our approximate filtering and smoothing methods, and marginal posterior means and standard deviations estimated using a Metropolis-within-Gibbs MCMC algorithm which incorporates our approximate marginal likelihood, under three sets of prior distributions over $(\beta,\lambda,\rho,\gamma)$ labelled `vague', `informative' and `noncentered'  by \cite{lekone2006statistical}. The basic reproduction number is $R_0=\beta/\gamma$. The results show accurate recovery of the true parameter values. 

\begin{table*}[httb!]
	\caption{Parameter estimation for synthetic data under the Ebola model using our EM and MCMC methods under three sets of prior distributions specified by \cite{lekone2006statistical}. For the MCMC results, the posterior means is reported as the point estimate and the numbers in parentheses are posterior standard deviations.\newline}
	\label{tab:ebola_synthetic}
	\centering
	\scriptsize
	\begin{tabular*}{1\linewidth}{llllllll}
		\hline
		\textbf{Parameter}     & $\mathbf{\beta}$ & $\mathbf{\lambda}$ & $\mathbf{\rho}$ & $\mathbf{\gamma} $ & $\mathbf{q^{(2,3)}}$ & $\mathbf{q^{(3,4)}}$ & $\mathbf{R_0}$  \\
		\hline \\
		True value        & 0.2 & 0.2 & 0.2 & 0.143 & 0.92 & 0.75 & 1.40 \\
		MLE (EM-alg.)       & 0.20 & 0.18 & 0.21 & 0.139 & 1.00 & 0.81 &  1.44   \\
		MCMC (vague)    & 0.23 (0.028)  & 0.21 (0.080) & 0.22 (0.076) & 0.173 (0.024) & 0.81 (0.140) & 0.66 (0.119) & 1.31 (0.088)  \\
		MCMC (infor.)   & 0.22 (0.020) &  0.22 (0.065) & 0.20 (0.035) & 0.162 (0.017) & 0.83 (0.130) & 0.67 (0.112) & 1.34 (0.082)   \\
		MCMC (noncent.) & 0.32 (0.048) &  0.35 (0.101) & 0.17 (0.031) & 0.256 (0.049) & 0.79 (0.147) & 0.64 (0.125) & 1.28 (0.084)   \\
		\hline
	\end{tabular*}
\end{table*} 

\begin{table*}[httb!]
	\caption{Parameter estimation for the real Ebola data. Numbers in parentheses in column 5 are standard errors, for all other columns they are posterior standard deviations. For columns 2,3,4,6 the parameter estimates are posterior means. For each of $\beta$, $\lambda$, and $1/\rho$ the pairs of estimates in column 1 were obtained from the respective bi-modal posteriors by applying $k$-means clustering, with $k=2$, to the MCMC output. \newline}
	\label{tab:ebola}
	\centering
	\tiny
	\begin{tabular*}{1\linewidth}{llllllll}
		\hline
		\textbf{Parameter}     & $\mathbf{\beta}$ & $\mathbf{\lambda}$ & $\mathbf{1 \slash \rho}$ & $\mathbf{1 \slash \gamma} $ & $\mathbf{q^{(2,3)}}$ & $\mathbf{q^{(3,4)}}$ & $\mathbf{R_0}$  \\       
		\hline       \\
		\makecell{Our MCMC method\\ vague prior} & \makecell{ 0.36 (0.049)\\0.22 (0.025) } & \makecell{ 0.32 (0.140)\\0.05 (0.008) } & 
		\makecell{ 10.39 (1.554) \\ 1.86 (0.487)}   & 6.17 (1.042)   & 0.44 (0.103)   & 0.36 (0.088) & \makecell{  2.18 (0.227) \\1.42 (0.102) }   \\            
		\makecell{Our MCMC method\\ informative prior} & 0.26 (0.033) & 0.12 (0.064) &   6.07 (1.919) & 6.86 (0.834)    &  0.50 (0.109)   & 0.41 (0.093) & \makecell{ 1.64 (0.696) }  \\ 
		\makecell{\citep{lekone2006statistical}\\ vague prior}      & 0.24 (0.020)              & 0.16 (0.009)              &  9.43 (0.620) & 5.71 (0.548) & \_            & \_            & 1.38 (0.127)   \\ 
		\makecell{\citep{lekone2006statistical}\\ informative prior } & 0.21 (0.017)              & 0.15 (0.010)              & 10.11 (0.713) & 6.52 (0.564)  & \_             & \_    & 1.36 (0.128)           \\                  	                          
		\makecell{ODE + least squares\\ \citep{chowell2004basic}}                     & 0.33 (0.006)              & 0.98 (unknown)            &  5.30 (0.230) & 5.61 (0.190) & \_           & \_             & 1.83 (0.060)  \\ 
		\makecell{ABC ABSEIR\\\citep{brownABC}}                                       & 0.3 (0.088)              & 0.36 (0.325)              &  7.91 (2.703) & 15.01 (32.863) & \_ &\_  & 3.66 (6.592)\\
		\hline
	\end{tabular*}
\end{table*}

\paragraph{Real data} We analyzed the same real Congo Ebola data as in \cite{lekone2006statistical}. Table \ref{tab:ebola} shows several interesting findings. 1) The results from our methods are generally closer to those from the Data Augmentation MCMC sampler of \cite{lekone2006statistical}  than those from the ABC method of \cite{brownABC}; the former targets the true posterior distribution whilst the latter does so only approximately. 2) Under the `vague' prior our method finds bi-modal posteriors for $\beta$, $\lambda$, and $1/\rho$. For $\beta$, one of the modes roughly matches the posterior mean obtained using \cite{lekone2006statistical} whilst the other is more similar to the least-squares estimate from \cite{chowell2004basic}; we conjecture that our MCMC sampler has better mixing than that of \cite{lekone2006statistical}, allowing it to find these two modes. 3) We can report estimates for $q^{(2,3)}$ and $q^{(3,4)}$, whilst the other methods do not. Figure \ref{fig:posteriors_ebola} shows posterior and posterior-predictive distributions for the counts of new infectives each day. The former estimates for the true numbers which gave rise to the under-reported data, whilst the latter shows coverage of the data hence a good model fit \citep{gelman1996posterior}.

\begin{figure*}[httb!]
	\centering
	\includegraphics[width=\columnwidth]{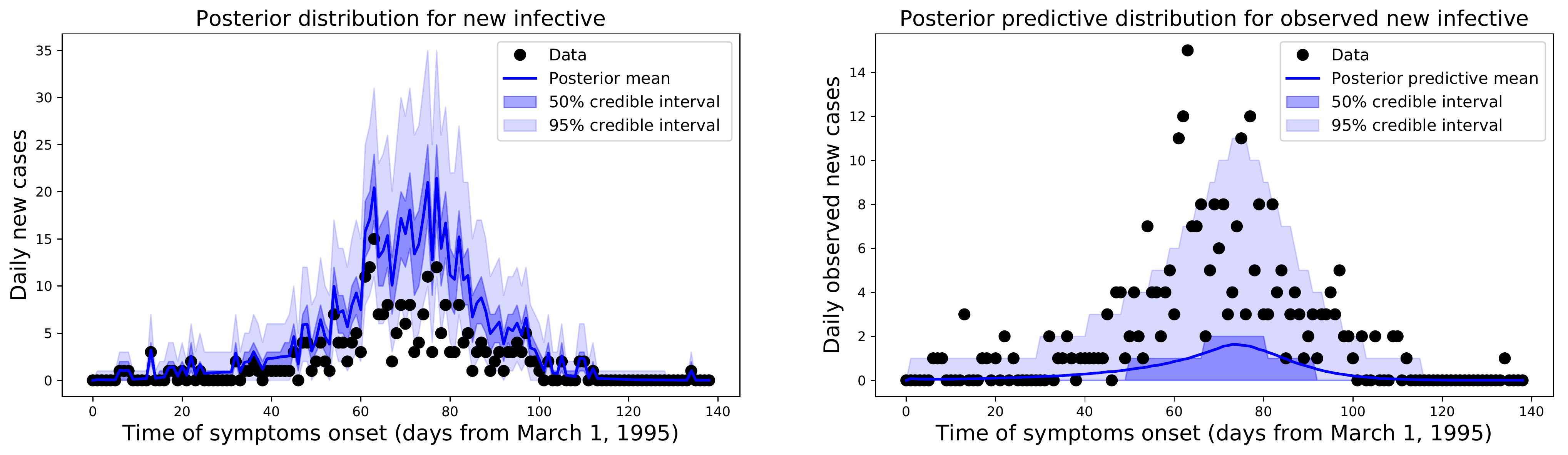}
	\caption{Analysis of real Ebola data with our method. Posterior smoothing distributions for the number of new infectives per day and posterior predictive distributions for the associated observations, i.e., subject to under-reporting. Control measures were introduced on day 70.} \label{fig:posteriors_ebola}
\end{figure*}

\begin{figure*}[httb!]
	\centering
	\includegraphics[width=\columnwidth]{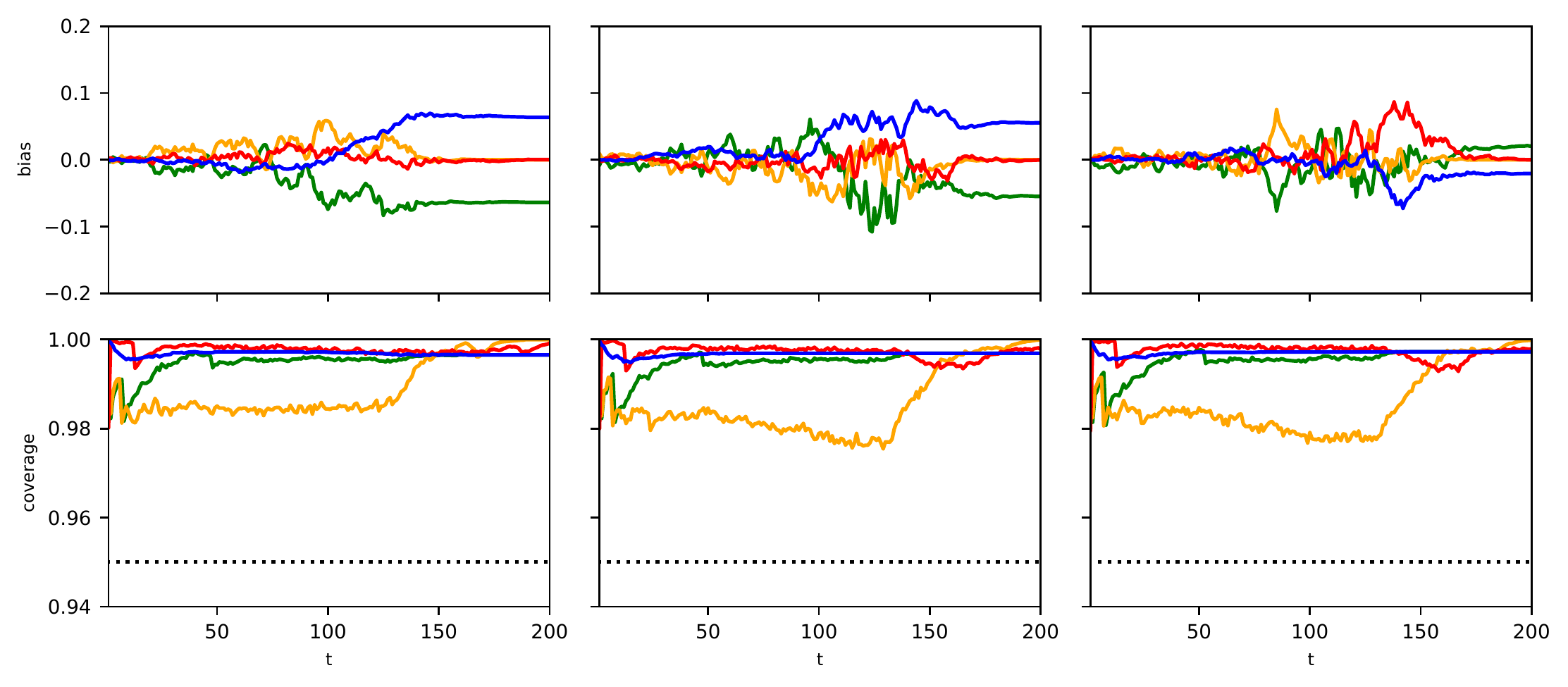}
	\caption{Empirical bias and empirical coverage of nominal $95\%$-credible intervals from $2\times10^4$  simulations over $200$ time steps of the Ebola model. Columns from left to right: $ n=5\times10^2, 5\times10^4, 5\times10^6 $. Top row: bias, bottom row: coverage. Red, yellow, blue, green correspond to $x_t^{(i)}$, $i=1,2,3,4$, i.e. susceptible, exposed, infective, recovered.}\label{fig:bias_and_cov}
\end{figure*}

\subsection{Accuracy: filtering bias and credible interval coverage}

The purpose of this subsection is to study the accuracy of the approximate filtering distributions obtained from algorithm \ref{alg:filtering_inc} when applied to the Ebola model described in subsection \ref{sec:ebola}. The ground truth parameter values $(\beta,\lambda,\rho,\gamma)$ in the synthetic data experiment were taken together with $q_t^{(2,3)}=291/316$, $q_t^{(3,4)}=236/316$. We considered three population sizes $n=5\times10^2, 5\times10^4, 5\times10^6 $, and in each case the initial distribution was $\bm{\pi}_0 = [1-1/n,1/n,0,0]^\mathrm{T}$. For each value of $n$, we simulated $2\times10^4$ data sets from the model, each over $200$ time steps.

To assess accuracy we considered bias and credible-interval coverage. For the former we calculated the empirical bias associated with the mean vector of the approximation to $p(\mathbf{x}_t|\mathbf{Y}_{1:t})$ obtained from algorithm \ref{alg:filtering_inc} as an estimator of $\mathbf{x}_t$. For the latter we calculated the empirical coverage of the nominal $95\%$-credible interval for the marginal over each $x_t^{(i)}$, $i=1,2,3,4$. For the true (i.e. approximation-free) filtering distributions, asymptotically in the number of simulated data sets the bias would be zero and the coverage would be $95\%$. 

Figure \ref{fig:bias_and_cov} shows that for all three values of $n$, the bias at every time step and for every compartment is less than $0.1$ in magnitude. This shows the approximation is very accurate: the true values of $x_t^{(i)}$, $i=1,2,3,4$ are always integers, and a bias less than $0.5$ in magnitude means that, on average, if the estimated number of individuals is rounded to the nearest integer, the true number of individuals is recovered.   The credible interval coverage reported in figure \ref{fig:bias_and_cov} shows that the approximate filtering distributions tend to over-represent uncertainty: the empirical coverage at all time steps for all compartments of the nominal $95\%$ interval is is between $97\%$ and $100\%$. The bias and coverage appear robust to population size.

\begin{figure*}[httb!]
	\centering
	\includegraphics[width=\columnwidth]{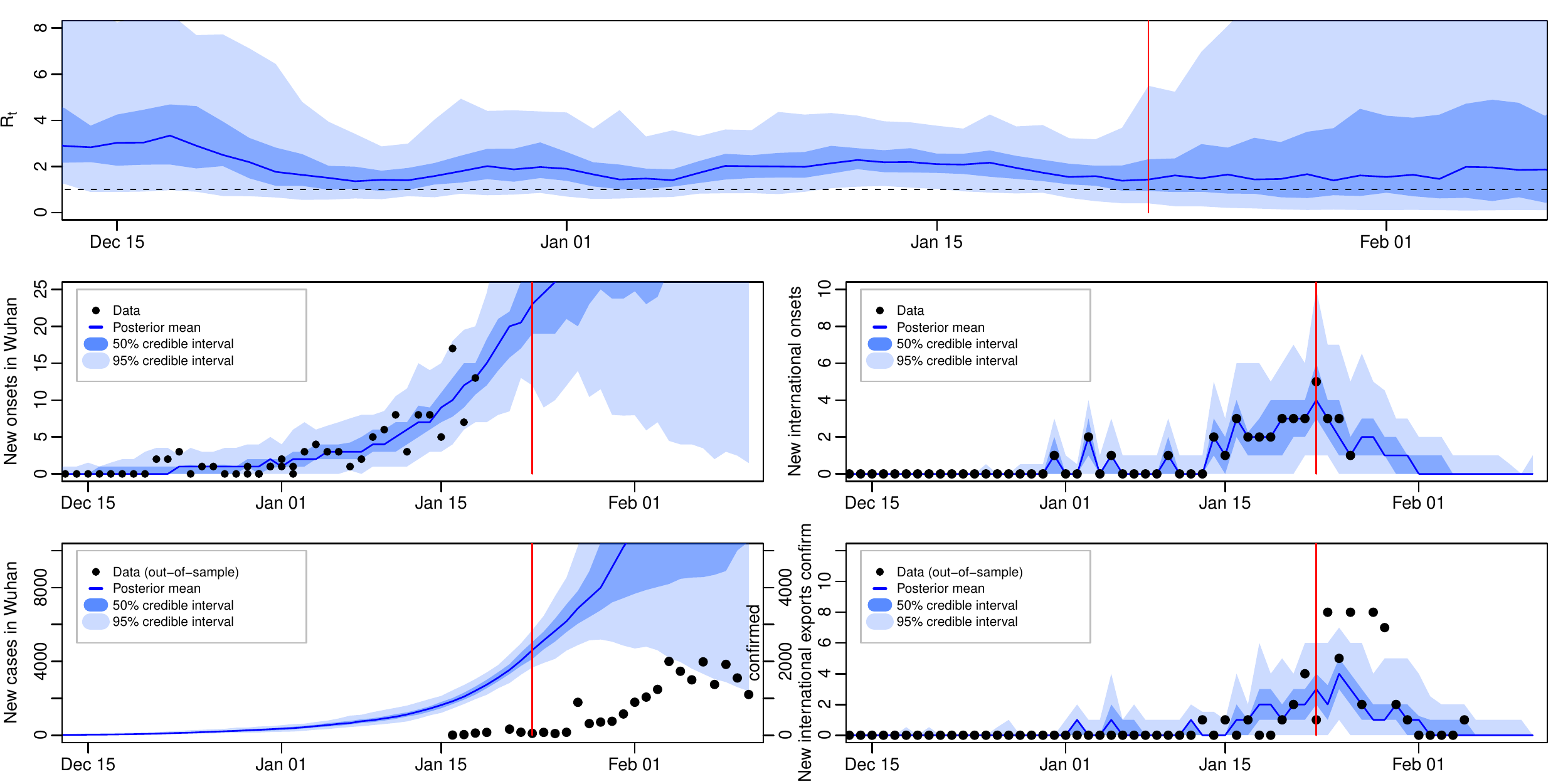}
	\caption{Results for the COVID-19 model using our methods.  Red line is date at which travel restrictions were introduced. Top: estimated reproduction number. Middle row:  estimated daily new confirmed cases in Wuhan (left) and internationally (right), both with in-sample data by date of symptom onset. Bottom row, left: estimated new symptomatic but possibly unconfirmed cases (left axis) and out-of-sample new confirmed cases data (right axis); right: estimated confirmed international cases by date of confirmation, and out-of-sample data.} \label{fig:covid}
\end{figure*}

\subsection{Estimating the time-varying reproduction number of COVID-19 in Wuhan, China}\label{sec:covid}
A compartmental model for estimating the time-varying reproduction number of COVID-19 in Wuhan, China, has recently been published in \cite{kucharski2020early}. The model has 15 compartments: susceptibles in Wuhan become exposed and either stay in Wuhan or depart internationally, then in either case pass through further stages being exposed, infective, symptomatic and confirmed.  The transmission rate is modelled as time-varying $(\beta_t)_{t\geq0}$,  a-priori by a geometric random walk, and $\beta_t$ is considered proportional to the reproductive number $R_t$. \cite{kucharski2020early} proposed a Sequential Monte Carlo (SMC) algorithm to estimate $(R_t)_{t\geq 0}$ which weights samples of $(\beta_t)_{t\geq0}$ by the likelihood of the associated ODE solution under a Poisson observation model. Our methods can be used to replace their ODE model with a discrete-time stochastic version of the compartmental model, and with their Poisson model replaced our binomial observation model from section \ref{sec:obs_model_z}. Our version of the SMC algorithm weights samples of $(\beta_t)_{t\geq0}$ by their approximate marginal likelihoods, computed using our multinomial filtering techniques. We jointly analyzed two of three data sets from \cite{kucharski2020early}: daily counts of new infectives by date of symptom onset in Wuhan, and internationally exported from Wuhan.

Figure \ref{fig:covid} shows our results in the format of \cite{kucharski2020early}. Compared to results obtained using their ODE model (see supplementary material), our estimates of $R_t$ are generally lower, and closer $1$ for the period after travel restrictions are introduced; and our credible intervals for the in-sample plots are generally wider, reflecting the stochastic nature of our compartmental model. In the bottom two plots our posteriors are mostly concentrated on lower values than those from the method of \cite{kucharski2020early}.



\subsubsection*{Acknowledgements}
Nick Whiteley was supported by a Turing Fellowship from the Alan Turing Institute, GW4 and Jean Golding Institute seed-corn funding and thanks Lawrence Murray for discussions about epidemic models. Lorenzo Rimella was supported the Alan Turing Institute PhD Enrichment Scheme. 

\bibliography{SEIRbiblio}
\bibliographystyle{plainnat}



\appendix

\maketitle

\section{Reader's guide to the appendices}
Appendix \ref{sec:stoch_vs_determ} discusses some of the shortcomings of deterministic compartmental models, expanding on the discussion in section \ref{subsec:determ_and_stoch_models}. Appendix \ref{sec:filt_and_smooth} contains proofs of lemmas \ref{lem:prediction} and \ref{lem:update},  plus corresponding results and proofs for the observation model from section \ref{sec:obs_model_z}, and derivation of smoothing algorithms. Appendix \ref{sec:ebola_supp} provides additional details and numerical results for the Ebola example in section \ref{sec:ebola}. Additional details and numerical results for the COVID-19 example in section \ref{sec:covid} are given in appendix \ref{sec:covid_sup}.

\section{Stochastic vs. deterministic SEIR models}\label{sec:stoch_vs_determ}
Perhaps the most widely applied formulation of a compartmental model is as a system of ordinary differential equations. 

\paragraph{SEIR example.} For a population of size $n$, the SEIR ODE model is:
\begin{equation}
\frac{\mathrm{d}S}{\mathrm{d} t} = -  \frac{\beta S I}{n},\qquad \frac{\mathrm{d}E}{\mathrm{d} t} =  \frac{\beta S I}{n} - \rho E ,\qquad \frac{\mathrm{d}I}{\mathrm{d} t} = \rho E -\gamma I , \qquad \frac{\mathrm{d}R}{\mathrm{d} t} = \gamma I,\label{eq:SEIR_ode}
\end{equation}
initialized with nonnegative integers in each of the compartments $(S_0, E_0, I_0, R_0)$ such that $S_0+E_0+I_0+R_0=n$.

The most obvious drawback of ODE models is that, once model parameters and the initial population are fixed, any discrepancy between observed data and the solution of the ODE has to be explained as observation error, which is a serious restriction from a modelling point of view. In practice one can try to estimate unknown parameters and/or the initial condition by numerically minimizing this discrepancy, e.g. under squared error loss. Standard errors for parameter estimates can be derived using asymptotic theory for nonlinear least squares, but calculation of them in practice involves numerical differentiation of the ODE solution flow w.r.t. parameters  \citep{chowell2004basic}. When a probabilistic observation model is specified, Bayesian approaches allow for uncertainty quantification over parameters via posterior distributions, but evaluating the likelihood function for model parameters still involves numerical solution of the ODE.

Figure \ref{fig:seir_sim} shows simulation output for the proportion of infective individuals in the discrete-time SEIR model with $n=10^2,10^3,10^5, 10^7$ and initial conditions $(S_0, E_0, I_0, R_0)=(n-1,1,0,0)$, with $\beta=0.8$, $\rho=1/9$ and $\gamma=0.2$. It is evident that the sample paths become smoother as $n$ grows, but there is still substantial variability across sample paths even with $n=10^7$. This can be explained by the fact that since $(E_0, I_0)=(1,0)$ independently of $n$, the numbers of exposed and infective individuals in the  first few time periods of the epidemic are typically very small, despite the fact that the overall population size may be large, and the statistical variability associated with these small numbers has a lasting effect on the overall timing of the outbreak.

To explain how this relates to ODE limits, for $n\geq1$ and initial proportions of the population $(s_0, e_0, i_0, r_0)$, i.e., $s_0+ e_0+ i_0 + r_0=1$, let us write $\mathcal{D}_n(s_0, e_0, i_0, r_0)$ for the collection of ODEs in (\ref{eq:SEIR_ode}) together with the initial condition  $(n s_0, n e_0, n i_0, n r_0)$. It can be checked by substitution that $t\mapsto (S_t, E_t, I_t, R_t)$ is a solution of $\mathcal{D}_1(s_0, e_0, i_0, r_0)$ if and only if $t\mapsto (n S_t, n E_t, n I_t, n R_t)$ is a solution of $\mathcal{D}_n(s_0, e_0, i_0, r_0)$. Thus we see that $n$ plays a trivial role in the ODE model: it is just a scaling factor for the solution.

The limit theorems of \cite{kurtz1970solutions, kurtz1971limit} applied in this situation pertain to the probabilistic convergence on a finite time-window as $n\to\infty$ of the path of the continuous-time SEIR Markov process with initial condition $(S_0, E_0, I_0, R_0)=(n-1,1,0,0)$ and compartment counts normalized by $n$, to the solution of $\mathcal{D}_1(s_0, e_0, i_0, r_0)$, where $(s_0, e_0, i_0, r_0) = \lim_{n\to\infty}(n-1,1,0,0)/n=(1,0,0,0)$. However  in the solution of $\mathcal{D}_1(1,0,0,0)$, the exposed and infective compartments are always empty, i.e. an epidemic never occurs. This can be reconciled with figure \ref{fig:seir_sim} by observing that there the peak in the number of infectives typically occurs later as $n$ grows: the limiting $n\to\infty$ case is that in which a peak \emph{never} occurs.

This illustrates that sequences of well-behaved stochastic models can have ODE limits which are unrealistic to the point of being pathological, and therefore these limits are not always a sensible justification for using ODE models.

\begin{figure}[httb!]
	\centering
	\includegraphics[width=\textwidth]{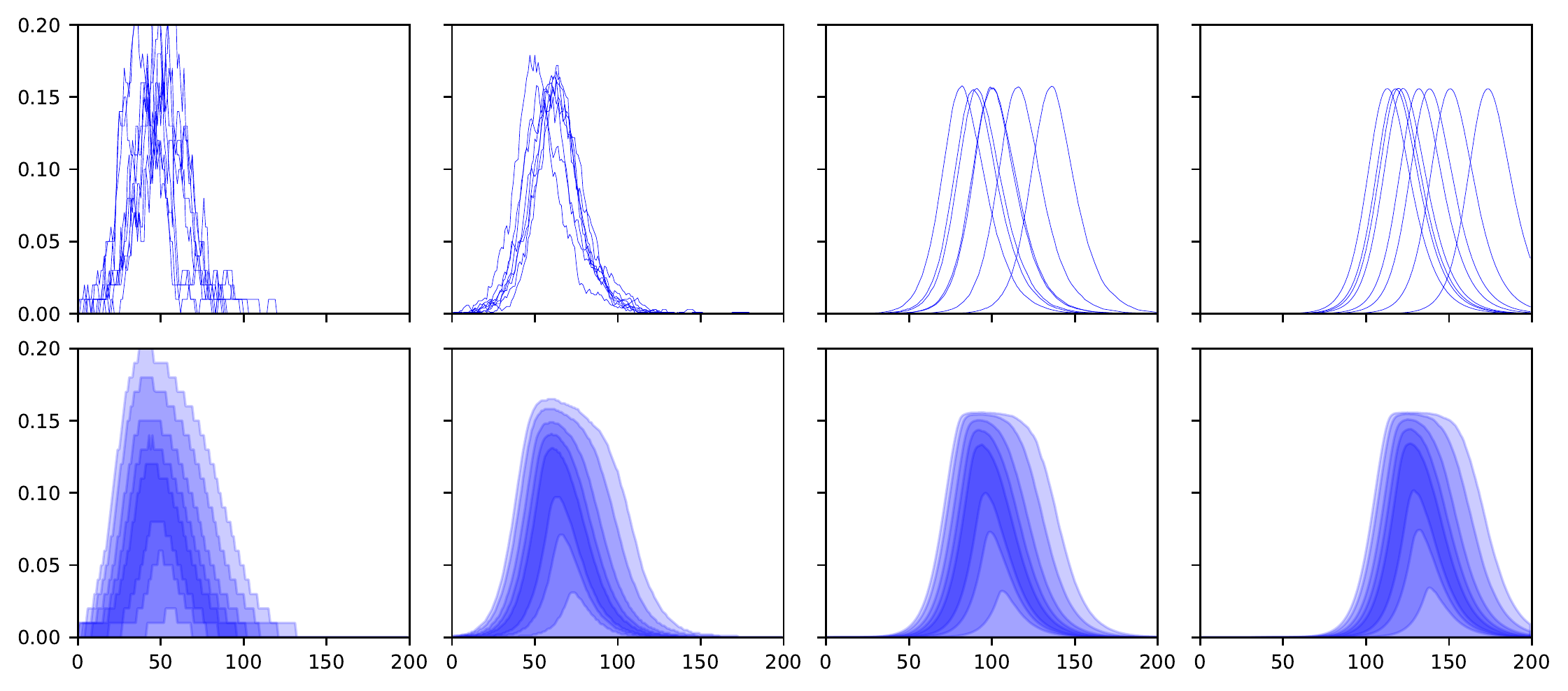}
	\caption{Proportion of infective individuals against time in simulations from the discrete-time stochastic SEIR model with $\beta=0.8$, $\rho=1/5$, $\gamma=1/9$ and $h=1$ over $200$ time steps. Columns left to right: $n=10^2,10^3,10^5, 10^7$ and initial conditions $(S_0, E_0, I_0, R_0)=(n-1,1,0,0)$. Top row: 10 realizations from the model. Bottom row: at each time step shaded regions indicate percentile intervals of the form $[\alpha,1-\alpha]\times 100\%$, for $\alpha=0.05,0.1,0.2,0.3,0.4$ estimated from $10^4$ realizations from the model. }\label{fig:seir_sim}
\end{figure}

\section{Supplementary information about multinomial filtering and smoothing}\label{sec:filt_and_smooth}

\subsection{Proofs of Lemmas \ref{lem:prediction} and \ref{lem:update}}
\proof[Proof of Lemma \ref{lem:prediction}]{
	Since $\mu$ is the probability mass function associated with $\mathrm{Mult}(n,\bm{\pi})$, $\mathbb{E}_{\mu}[\bm{\eta}(\mathbf{x})]=\bm{\pi}$, so we need to prove that $\sum_{\mathbf{x} \in \mathcal{S}_{m,n}} \mu(\mathbf{x}) M_t(\mathbf{x},\bm{\pi},\cdot)$ is the probability mass function associated with $\mathrm{Mult}(n,\bm{\pi}^{\mathrm{T}}\mathbf{K}_{t,\bm{\pi}})$. We shall achieve this using the unique characterization of the probability mass function by its moment generating function.

	With $\mathbf{x}\sim \mu$, let $\widetilde{\mathbf{x}}\sim M_t(\mathbf{x},\bm{\pi},\cdot)$, so by construction $\sum_{\mathbf{x} \in \mathcal{S}_{m,n}} \mu(\mathbf{x}) M_t(\mathbf{x},\bm{\pi},\cdot)$ is the marginal  probability mass function of $\widetilde{\mathbf{x}}$. Therefore by the definition of $M_t$ in section \ref{sec:approx_prediction},  $\widetilde{\mathbf{x}} = (\mathbf{1}_m^{\mathrm{T}} \mathbf{Z} ) ^{\mathrm{T}}$, where the rows of $\mathbf{Z}$ are conditionally  independent given $\mathbf{x}$, and the conditional distribution of  the  $i$th row of $\mathbf{Z}$ given $\mathbf{x}$ is $\mathrm{Mult}(x^{(i)},\mathbf{K}_{t,\bm{\pi}}^{(i,\cdot)})$. Using these facts, the moment generating function of $\widetilde{\mathbf{x}}$ can be written:
	\begin{align}
	\mathbb{E}\left[ \exp(\widetilde{\mathbf{x}}^{\mathrm{T}}\mathbf{b})\right] &=\mathbb{E}\left[\exp(\mathbf{1}_m^{\mathrm{T}}\mathbf{Z}\mathbf{b})\right]\nonumber\\
	&=\mathbb{E}\left[\exp\left(\sum_{i,j=1}^{m}Z^{(i,j)} b^{(j)}\right)\right]\nonumber\\
	&=\mathbb{E}\left[\mathbb{E}\left[\left.\exp\left(\sum_{i,j=1}^{m}Z^{(i,j)}b^{(j)}\right)\right|\mathbf{x}\right]\right]\nonumber \\
	&=\mathbb{E}\left[\prod_{i=1}^{m}\mathbb{E}\left[\left.\exp\left(\sum_{j=1}^mZ^{(i,j)}b^{(j)}\right)\right|x^{(i)}\right]\right].\label{eq:mgf_proof}
	\end{align}
	Now we use the fact that, again by definition of $M_{t}$, $\mathbb{E}\left[\left.\exp\left(\sum_{j=1}^{m}Z^{(i,j)}b^{(j)}\right)\right|x^{(i)}\right]$
	is the m.g.f. of $\text{Mult}(x^{(i)},\mathbf{K}_{t,\bm{\pi}}^{(i,\cdot)})$,
	where $\mathbf{K}_{t,\bm{\pi}}$ has elements $k_{t,\bm{\pi}}^{(i,j)}$,
	i.e. :
	\begin{equation*}
	\mathbb{E}\left[\left.\exp\left(\sum_{j=1}^{m}Z^{(i,j)}b^{(j)}\right)\right|x^{(i)}\right]=\left(\sum_{j=1}^{m}k_{t,\bm{\pi}}^{(i,j)}e^{b^{(j)}}\right)^{x^{(i)}}.
	\end{equation*}
	Substituting this into (\ref{eq:mgf_proof}) and then using
	$\mathbf{x}\sim\mu=\text{Mult}(n,\bm{\mathbf{\pi}})$, 
	\begin{align*}
	\mathbb{E}\left[ \exp(\widetilde{\mathbf{x}}^{\mathrm{T}}\mathbf{b})\right]&= \mathbb{E}\left[\prod_{i=1}^{m}\mathbb{E}\left[\left.\exp\left(\sum_{j=1}^{m}Z^{(i,j)}b^{(j)}\right)\right|x^{(i)}\right]\right] =\mathbb{E}\left[\prod_{i=1}^{m}\left(\sum_{j=1}^{m}k_{t,\bm{\pi}}^{(i,j)}e^{b^{(j)}}\right)^{x^{(i)}}\right]\\
	& =\sum_{(x^{(1)},\ldots,x^{(m)})\in\mathcal{S}_{m,n}}n!\prod_{i=1}^{m}\frac{(\pi^{(i)})^{x^{(i)}}}{x^{(i)}!}\left(\sum_{j=1}^{m}k_{t,\bm{\pi}}^{(i,j)}e^{b^{(j)}}\right)^{x^{(i)}}\\
	& =\sum_{(x^{(1)},\ldots,x^{(m)})\in\mathcal{S}_{m,n}}n!\prod_{i=1}^{m}\frac{1}{x^{(i)}!}\left(\sum_{j=1}^{m}\pi^{(i)}k_{t,\bm{\pi}}^{(i,j)}e^{b^{(j)}}\right)^{x^{(i)}}\\
	& =\left(\sum_{i=1}^{m}\sum_{j=1}^{m}\pi^{(i)}k_{t,\bm{\pi}}^{(i,j)}e^{b^{(j)}}\right)^{n}= \left(\sum_{j=1}^m (\bm{\pi}^{\mathrm{T}}\mathbf{K}_{t,\bm{\pi}})^{(j)}e^{b^{(j)}}\right)^{n},
	\end{align*}
	where the penultimate equality holds by the multinomial theorem. The proof is completed by noticing that $\left(\sum_{j=1}^m (\bm{\pi}^{\mathrm{T}}\mathbf{K}_{t,\bm{\pi}})^{(j)}e^{b^{(j)}}\right)^{n}$
	is the moment generating function of $\text{Mult}(n,\bm{\pi}^{\mathrm{T}}\mathbf{K}_{t,\bm{\pi}})$.
	\qed}

\proof[Proof of Lemma \ref{lem:update}]{
	Under the distributional assumptions in the statement of the lemma, for $\mathbf{x}\in\mathcal{S}_{m,n}$,
	\begin{equation*}
	p(\mathbf{x})=\frac{n!}{\prod_{j=1}^{m}x^{(j)}!}\prod_{j=1}^{m}(\pi^{(j)})^{x^{(j)}},
	\end{equation*}
	and for $0\leq y^{(j)}\leq x^{(j)},\,j=1,\ldots,m$,
	\begin{equation*}
	p(\mathbf{y}|\mathbf{x})=\prod_{j=1}^{m}\frac{x^{(j)}!}{y^{(j)}!(x^{(j)}-y^{(j)})!}(q^{(j)})^{y^{(j)}}(1-q^{(j)})^{x^{(j)}-y^{(j)}}.
	\end{equation*}
	Therefore
	\begin{equation}
	p(\mathbf{x},\mathbf{y})=n!\prod_{j=1}^{m}\frac{(\pi^{(j)})^{x^{(j)}}(q^{(j)})^{y^{(j)}}(1-q^{(j)})^{x^{(j)}-y^{(j)}}}{y^{(j)}!(x^{(j)}-y^{(j)})!},\label{eq:p(x,y)}
	\end{equation}
	and
	\begin{align}
	p(\mathbf{y}) & =\sum_{\{x^{(j)}:x^{(j)}\geq y^{(j)},\sum_{j}x^{(j)}=n\}}n!\prod_{j=1}^{m}\frac{(\pi^{(j)})^{x^{(j)}}(q^{(j)})^{y^{(j)}}(1-q^{(j)})^{x^{(j)}-y^{(j)}}}{y^{(j)}!(x^{(j)}-y^{(j)})!}\nonumber \\
	& =n!\left(\prod_{j=1}^{m}\frac{(q^{(j)})^{y^{(j)}}(\pi^{(j)})^{y^{(j)}}}{y^{(j)}!}\right)\nonumber\\
	&\qquad \times \sum_{\{x^{(j)}:x^{(j)}\geq y^{(j)},\sum_{j}x^{(j)}=n\}}\left(\prod_{j=1}^{m}\frac{(\pi^{(j)})^{x^{(j)}-y^{(j)}}(1-q^{(j)})^{x^{(j)}-y^{(j)}}}{(x^{(j)}-y^{(j)})!}\right)\nonumber \\
	& =n!\left(\prod_{j=1}^{m}\frac{(q^{(j)})^{y^{(j)}}(\pi^{(j)})^{y^{(j)}}}{y^{(j)}!}\right)\frac{\left(\sum_{j=1}^{m}\pi^{(j)}(1-q^{(j)})\right)^{n-\sum_{j=1}^{m}y^{(j)}}}{(n-\sum_{j=1}^{m}y^{(j)})!},\label{eq:p(y)}
	\end{align}
	where the final equality holds by the multinomial theorem. 
	Dividing
	(\ref{eq:p(x,y)}) by (\ref{eq:p(y)}) gives
	\begin{align*}
	p(\mathbf{x}|\mathbf{y}) & =\left(n-\sum_{j=1}^{m}y^{(j)}\right)!\prod_{j=1}^{m}\frac{1}{(x^{(j)}-y^{(j)})!}\left(\frac{\pi^{(j)}(1-q^{(j)})}{\sum_{i=1}^{m}\pi^{(i)}(1-q^{(i)})}\right)^{x^{(j)}-y^{(j)}},
	\end{align*}
	which is the probability mass function of $\mathbf{y}+\mathbf{x}^{\star}$
	given in the statement of the lemma.
	\qed}

\begin{remark}
	The probability mass function in (\ref{eq:p(y)}) has the interpretation of being a multinomial distribution over $n+1$ compartments, where the count variable associated with  the $(m+1)$th compartment is $n-\sum_{j=1}^{m}y^{(j)}$.
\end{remark}

\subsection{Approximating the prediction and update operations for $(\mathbf{Z}_t,\mathbf{Y}_t)_{t\geq 1}$}\label{sec:Z_approximations}
Algorithm \ref{alg:filtering_inc} is derived from lemma \ref{lem:Z_predict} and \ref{lem:Z_update}. Given $\mathbf{Z}$ and $\bm{\eta}$, let $\overline{M}_t(\mathbf{Z},\bm{\eta},\cdot)$ be the probability mass function of a random $m
\times m $ matrix, say $\widetilde{\mathbf{Z}}$, such that $\mathbf{1}_m^\mathrm{T} \mathbf{Z} = (\widetilde{\mathbf{Z}}\mathbf{1}_m)^\mathrm{T}$ with probability $1$, and such that given the row-sums $\widetilde{\mathbf{Z}}\mathbf{1}_m$, rows of $\widetilde{\mathbf{Z}}$ are independent and the conditional distribution of the $i$th row is $\mathrm{Mult}(x^{(i)}, \mathbf{K}_{t,\bm{\eta}}^{(i,\cdot)}  )$ where $\mathbf{x}^{\mathrm{T}}=\mathbf{1}_m^{\mathrm{T}} \mathbf{Z}$. So by construction $\overline{M}_t(\mathbf{Z}_{t-1},\bm{\eta}(\mathbf{1}_{m}^{\mathrm{T}}\mathbf{Z}_{t-1}),\mathbf{Z}_{t})$ gives the transition probabilities of the Markov chain $(\mathbf{Z}_{t})_{t\geq1}$ defined in section \ref{sec:individ}.
\begin{lemma}\label{lem:Z_predict}
	If for a given $m\times m$ matrix $\mathbf{P}$, $\overline{\mu}$
	is the probability mass function associated with $\mathrm{Mult}(n,\mathbf{P})$
	and $\mathbb{E}_{\overline{\mu}}[\bm{\eta}(\mathbf{1}_{m}^{\mathrm{T}}\mathbf{Z})]$
	is the expected value of $\mathbf{1}_{m}^{\mathrm{T}}\mathbf{Z}$
	when $\mathbf{Z}\sim\mathrm{Mult}(n,\mathbf{P})$, then $\sum_{\mathbf{Z}}\overline{\mu}(\mathbf{Z})\overline{M}_t(\mathbf{Z},\mathbb{E}_{\overline{\mu}}[\bm{\eta}(\mathbf{1}_{m}^{\mathrm{T}}\mathbf{Z})],\cdot)$
	is the probability mass function associated with $\mathrm{Mult}(n,(\bm{\pi}\otimes\mathbf{1}_{m})\circ\mathbf{K}_{t,\bm{\pi}})$
	where $\bm{\pi}^{\mathrm{T}}=\mathbf{1}_{m}^{\mathrm{T}}\mathbf{P}$.
\end{lemma}

\begin{proof}
	The proof is similar to the proof of Lemma 1, so some steps and commentary
	are omitted. Note $\mathbb{E}_{\overline{\mu}}[\bm{\eta}((\mathbf{1}_{m}^{\mathrm{T}}\mathbf{Z})^{\mathrm{T}})]=(\mathbf{1}_{m}^{\mathrm{T}}\mathbf{P})^{\mathrm{T}}=\bm{\pi}$,
	and let $\widetilde{\mathbf{Z}}\sim\overline{M}(\mathbf{Z},\bm{\pi},\cdot)$.
	The moment generating function of $\mathbf{\widetilde{Z}}$ is:
	\begin{align*}
	\mathbb{E}[\exp(\mathbf{1}_{m}^{\mathrm{T}}(\mathbf{\widetilde{Z}}\circ\mathbf{B})\mathbf{1}_{m})] & =\mathbb{E}\left[\prod_{i=1}^{m}\exp\left(\sum_{j=1}^{m}\widetilde{z}^{(i,j)}b^{(i,j)}\right)\right]\\
	& =\mathbb{E}\left[\mathbb{E}\left[\left.\prod_{i=1}^{m}\exp\left(\sum_{j=1}^{m}\widetilde{z}^{(i,j)}b^{(i,j)}\right)\right|\mathbf{Z}\right]\right]\\
	& =\mathbb{E}\left[\prod_{i=1}^{m}\mathbb{E}\left[\left.\exp\left(\sum_{j=1}^{m}\widetilde{z}^{(i,j)}b^{(i,j)}\right)\right|(\mathbf{1}_{m}^{\mathrm{T}}\mathbf{Z})^{(i)}\right]\right]\\
	& =\mathbb{E}\left[\prod_{i=1}^{m}\left(\sum_{j=1}^{m}k_{t,\bm{\pi}}^{(i,j)}b^{(i,j)}\right)^{(\mathbf{1}_{m}^{\mathrm{T}}\mathbf{Z})^{(i)}}\right]\\
	& =n!\sum_{(x^{(1)},\ldots,x^{(m)})\in\mathcal{S}_{m,n}}\prod_{i=1}^{m}\frac{(\pi^{(i)})^{x^{(i)}}}{x^{(i)}!}\left(\sum_{j=1}^{m}k_{t,\bm{\pi}}^{(i,j)}e^{b^{(i,j)}}\right)^{x^{(i)}}\\
	& =\left(\sum_{i=1}^{m}\sum_{j=1}^{m}\pi^{(i)}k_{t,\bm{\pi}}^{(i,j)}e^{b^{(i,j)}}\right)^{n}.
	\end{align*}
\end{proof}
\begin{lemma}\label{lem:Z_update}
	Suppose that $\mathbf{Z}\sim\mathrm{Mult}(n,\mathbf{P})$ for a given
	$m\times m$ matrix $\mathbf{P}$, and given $\mathbf{Z}$, $\mathbf{Y}$
	is a matrix with conditionally independent entries distributed: $y^{(i,j)}\sim\mathrm{Bin}(z^{(i,j)},q^{(i,j)})$.
	Then the conditional distribution of $\mathbf{Z}$ given $\mathbf{Y}$
	is equal to that of $\mathbf{Y}+\mathbf{Z}^{\star}$ where
	\begin{equation*}
	\mathbf{Z}^{\star}\sim\mathrm{Mult}\left(n-\mathbf{1}_{m}^{\mathrm{T}}\mathbf{Y}\mathbf{1}_{m},\frac{\mathbf{P}\circ(\mathbf{1}_{m}\otimes\mathbf{1}_{m}-\mathbf{Q})}{1-\mathbf{1}_{m}^{\mathrm{T}}(\mathbf{P}\circ\mathbf{Q})\mathbf{1}_{m}}\right)
	\end{equation*}
	and 
	\begin{equation*}
	\mathbb{E}\left[\mathbf{Z}|\mathbf{Y}\right]= \mathbf{Y} + (n-\mathbf{1}_{m}^{\mathrm{T}}\mathbf{Y}\mathbf{1}_{m})\frac{\mathbf{P}\circ(\mathbf{1}_{m}\otimes\mathbf{1}_{m}-\mathbf{Q})}{1-\mathbf{1}_{m}^{\mathrm{T}}(\mathbf{P}\circ\mathbf{Q})\mathbf{1}_{m}}.
	\end{equation*}
	Moreover
	\begin{eqnarray*}
		\log p(\mathbf{Y})=\log(n!) + \mathbf{1}_m^{\mathrm{T}} (\mathbf{Y} \circ \log \mathbf{P} )\mathbf{1}_m + \mathbf{1}_m^{\mathrm{T}} (\mathbf{Y} \circ \log \mathbf{Q})\mathbf{1}_m -  \mathbf{1}_m^{\mathrm{T}} \log(\mathbf{Y} !)\mathbf{1}_m \\
		\qquad+ (n-\mathbf{1}_m^{\mathrm{T}} \mathbf{Y} \mathbf{1}_m) \log (1 -\mathbf{1}_m^{\mathrm{T}} (\mathbf{P} \circ \mathbf{Q}) \mathbf{1}_m ) - \log ((n-\mathbf{1}_m^{\mathrm{T}} \mathbf{Y} \mathbf{1}_m)!).
	\end{eqnarray*}
\end{lemma}
The proof is very similar to the proof of Lemma 2 so is omitted.

\subsection{Smoothing with the observation model from section \ref{sec:obs_model_x}}\label{sec:smoothing_x}
Assuming that algorithm \ref{alg:filtering_comp} has already been run up to a given time $t$,
our objective in this section is to derive algorithm \ref{alg:smoothing_x}, from which we take the approximations:
\begin{equation*}
p(\mathbf{x}_s|\mathbf{y}_{1:t})\approx \mathrm{Mult}(n,\bm{\pi}_{s|t}),\qquad p(\mathbf{Z}_s|\mathbf{y}_{1:t})\approx \mathrm{Mult}(n,(\mathbf{1}_m\otimes\bm{\pi}_{s|t})\circ \mathbf{L}_{s-1}^\mathrm{T}).
\end{equation*}

\begin{algorithm} 
	\caption{Multinomial smoothing with observations derived from $(\mathbf{x}_t)_{t\geq1}$}\label{alg:smoothing_x}
	\begin{algorithmic}[1]
		\For{$s=t-1,\ldots, 0$}
		\State let $\mathbf{L}_s$ be the matrix with elements $l_s^{(i,j)}\leftarrow\pi_{s|s}^{(j)}\, k_{s+1,\bm{\pi}_{s|s}}^{(j,i)}/(\bm{\pi}_{s|s}^\mathrm{T}\mathbf{K}_{s+1,\bm{\pi}_{s|s}})^{(i)}$
		\State $\bm{\pi}_{s|t} \leftarrow ( \bm{\pi}_{s+1|t}^{\mathrm{T}}\mathbf{L}_s)^{\mathrm{T}}  $
		\EndFor
	\end{algorithmic}
\end{algorithm}

We start by considering the identities:
\begin{align}
& p(\mathbf{x}_{0:t}|\mathbf{y}_{1:t})=p(\mathbf{x}_{t}|\mathbf{y}_{1:t})\prod_{s=0}^{t-1}p(\mathbf{x}_{s}|\mathbf{x}_{s+1},\mathbf{y}_{1:s})\nonumber\\
& p(\mathbf{x}_{s}|\mathbf{x}_{s+1},\mathbf{y}_{1:s})= \frac{p(\mathbf{x}_{s}|\mathbf{y}_{1:s})p(\mathbf{x}_{s+1}|\mathbf{x}_{s})}{p(\mathbf{x}_{s+1}|\mathbf{y}_{1:s})}\label{eq:backward_kernel_x}\\
&p(\mathbf{x}_{s+1}|\mathbf{y}_{1:s}) = \sum_{\mathbf{x}_s\in\mathcal{S}_{m,n}} p(\mathbf{x}_{s}|\mathbf{y}_{1:s})p(\mathbf{x}_{s+1}|\mathbf{x}_{s}),\nonumber
\end{align}
with the conventions $p(\mathbf{x}_{0}|\mathbf{x}_{1},\mathbf{y}_{1:0})\equiv p(\mathbf{x}_{0}|\mathbf{x}_{1})$,
$p(\mathbf{x}_{0}|\mathbf{y}_{1:0})\equiv p(\mathbf{x}_{0})$. The smoothing distributions $p(\mathbf{x}_{s}|\mathbf{y}_{1:t})$, $s=0,\ldots,t$, satisfy the `backward' recursion.
\begin{equation}
\sum_{\mathbf{x}_{s+1}\in\mathcal{S}_{m,n}}p(\mathbf{x}_{s+1}|\mathbf{y}_{1:t}) p(\mathbf{x}_{s}|\mathbf{x}_{s+1},\mathbf{y}_{1:s})= p(\mathbf{x}_{s}|\mathbf{y}_{1:t}). \label{eq:backward_recursion}
\end{equation}
All these formulae are standard identities for hidden Markov models \citep{briers2010smoothing}. 

In order to approximate (\ref{eq:backward_recursion}) we approximate each of the terms $p(\mathbf{x}_{s}|\mathbf{x}_{s+1},\mathbf{y}_{1:s})$. Consider first the numerator in (\ref{eq:backward_kernel_x}). Recall from section \ref{sec:approx_prediction} that the transition probabilities of the $(\mathbf{x}_s)_{s\geq0}$ process can be written in terms of $(M_{s})_{s\geq1}$, so:
\begin{equation}
p(\mathbf{x}_{s}|\mathbf{y}_{1:s})p(\mathbf{x}_{s+1}|\mathbf{x}_{s})=p(\mathbf{x}_{s}|\mathbf{y}_{1:s})M_{s+1}(\mathbf{x}_{s},\bm{\eta}(\mathbf{x}_{s}),\mathbf{x}_{s+1}).\label{eq:backward_kernel_num}
\end{equation}
We replace $p(\mathbf{x}_{s}|\mathbf{y}_{1:s})$ in (\ref{eq:backward_kernel_num}) by its multinomial approximation $\mathrm{Mult}(n,\bm{\pi}_{s|s})$ obtained using algorithm \ref{alg:filtering_comp}, 
and replace $\bm{\eta}(\mathbf{x}_{s})$ in  (\ref{eq:backward_kernel_num}) by its expected value under this multinomial distribution, i.e. $\bm{\pi}_{s|s}$, to give
\begin{equation}
p(\mathbf{x}_{s}|\mathbf{x}_{s+1},\mathbf{y}_{1:s})\approx\frac{\mu_{s|s}(\mathbf{x}_{s})M_{s+1}(\mathbf{x}_{s},\bm{\pi}_{s|s},\mathbf{x}_{s+1})}{\sum_{\mathbf{z}\in\mathcal{S}_{m,n}}\mu_{s|s}(\mathbf{z})M_{s+1}(\mathbf{z},\bm{\pi}_{s|s},\mathbf{x}_{s+1})},\label{eq:backward_appr_x}
\end{equation}
where $\mu_{s|s}(\cdot)$ is the probability mass function associated with $\mathrm{Mult}(n,\bm{\pi}_{s|s})$.

\begin{lemma}\label{lem:smoothing_x_kernel}
	With $\mathbf{x}_{s+1}$ considered fixed, the function which maps $\mathbf{x}_{s}$ to the right hand side of
	(\ref{eq:backward_appr_x}) is the probability mass function associated
	with $\mathbf{1}_{m}^{\mathrm{T}}\widetilde{\mathbf{Z}}$, where $\widetilde{\mathbf{Z}}$ is an $m\times m$ random matrix whose $i$th row
	has distribution $\mathrm{Mult}(x_{s+1}^{(i)},\mathbf{L}_{s}^{(i,\cdot)})$
	and $\mathbf{L}_{s}$ is the row-stochastic matrix with entries
	\begin{equation*}
	l_s^{(i,j)}=\frac{\pi_{s|s}^{(j)} \,k_{s+1,\bm{\pi}_{s|s}}^{(j,i)}}{(\bm{\pi}_{s|s}^\mathrm{T}\mathbf{K}_{s+1,\bm{\pi}_{s|s}})^{(i)}},
	\end{equation*}
	where $k_{s+1,\bm{\pi}_{s|s}}^{(i,j)}$ are the elements of $\mathbf{K}_{s+1,\bm{\pi}_{s|s}}$.
\end{lemma}
\begin{proof}
	Recalling the definition of $M_{s+1}$ from section \ref{sec:approx_prediction}, the numerator in (\ref{eq:backward_appr_x}) is:
	\begin{align}
	&\mu_{s|s}(\mathbf{x}_{s})M_{s+1}(\mathbf{x}_{s},\bm{\pi}_{s|s},\mathbf{x}_{s+1})\nonumber\\
	& =\left(n!\prod_{i=1}^{m}\frac{(\pi_{s|s}^{(i)})^{x_{s}^{(i)}}}{x_{s}^{(i)}!}\right)\left(\sum_{\mathbf{Z}\in\mathcal{T}_{m,n}(\mathbf{x}_{s},\mathbf{x}_{s+1})}\prod_{i=1}^{m}x_{s}^{(i)}!\prod_{j=1}^{m}\frac{(k_{s+1,\bm{\pi}_{s|s}}^{(i,j)})^{z^{(i,j)}}}{z^{(i,j)}!}\right)\nonumber\\
	& =n!\left(\prod_{i=1}^m (\pi_{s|s}^{(i)})^{x_{s}^{(i)}}\right)\sum_{\mathbf{Z}\in\mathcal{T}_{m,n}(\mathbf{x}_{s},\mathbf{x}_{s+1})}\prod_{i,j=1}^{m}\frac{(k_{s+1,\bm{\pi}_{s|s}}^{(i,j)})^{z^{(i,j)}}}{z^{(i,j)}!}\nonumber\\
	& =n!\sum_{\mathbf{Z}\in\mathcal{T}_{m,n}(\mathbf{x}_{s},\mathbf{x}_{s+1})}\prod_{i,j=1}^{m}\frac{(\pi_{s|s}^{(i)})^{z^{(i,j)}}(k_{s+1,\bm{\pi}_{s|s}}^{(i,j)})^{z^{(i,j)}}}{z^{(i,j)}!},\label{eq:smooth_proof_num}
	\end{align}
	where $\mathcal{T}_{m,n}(\mathbf{x}_{s},\mathbf{x}_{s+1})$ is the set of $m\times m$ matrices, say $\mathbf{Z}$, with nonnegative entries such that $\mathbf{Z}\mathbf{1}_{m}=\mathbf{x}_{s}$,  $(\mathbf{1}_{m}^{\mathrm{T}}\mathbf{Z})=\mathbf{x}_{s+1}^{\mathrm{T}}$, and $\mathbf{1}_m^{\mathrm{T}}\mathbf{Z}\mathbf{1}_m=n$.
	
	Now in order to derive an expression for the the denominator in (\ref{eq:backward_appr_x}), observe that (\ref{eq:smooth_proof_num}) can be disintegrated to give:
	\begin{equation*}
	n! \prod_{i,j=1}^{m}\frac{(\pi_{s|s}^{(i)})^{z^{(i,j)}}(k_{s+1,\bm{\pi}_{s|s}}^{(i,j)})^{z^{(i,j)}}}{z^{(i,j)}!},
	\end{equation*}
	which is the probability mass function of $\mathbf{Z}\sim\mathrm{Mult}(n,(\bm{\pi}_{s|s}\otimes\mathbf{1}_m)\circ\mathbf{K}_{s+1,\bm{\pi}_{s|s}})$. Therefore using the fact the marginal of this multinomial distribution over $\mathbf{1}_m^\mathrm{T} \mathbf{Z}$ is $\mathrm{Mult}(n,\bm{\pi}_{s|s}^{\mathrm{T}}\mathbf{K}_{s+1,\bm{\pi}_{s|s}})$, the denominator in (\ref{eq:backward_appr_x}) is
	\begin{equation}
	\sum_{\mathbf{x}_s\in\mathcal{S}_{m,n}}\mu_{s|s}(\mathbf{x}_s)M_{s+1}(\mathbf{x}_s,\bm{\pi}_{s|s},\mathbf{x}_{s+1})
	=n! \prod_{j=1}^m \frac{((\bm{\pi}_{s|s}^\mathrm{T}\mathbf{K}_{s+1,\bm{\pi}_{s|s}})^{(j)})^{x_{s+1}^{(j)}}}{x_{s+1}^{(j)}!}.\label{eq:smooth_proof_denom}
	\end{equation}
	Dividing (\ref{eq:smooth_proof_num}) by (\ref{eq:smooth_proof_denom}) gives:
	\begin{equation*}
	\sum_{\mathbf{Z}\in\mathcal{T}_{m,n}(\mathbf{x}_{s},\mathbf{x}_{s+1})}\prod_{i=1}^{m}x_{s+1}^{(j)}!\prod_{j=1}^{m}\left(\frac{\pi_{s|s}^{(i)} \,k_{s+1,\bm{\pi}_{s|s}}^{(i,j)}}{(\bm{\pi}_{s|s}^\mathrm{T}\mathbf{K}_{s+1,\bm{\pi}_{s|s}})^{(j)}}\right)^{z^{(i,j)}}\frac{1}{z^{(i,j)}!}.
	\end{equation*}
	Re-writing this sum with the change of variable $\widetilde{\mathbf{Z}}\coloneqq\mathbf{Z}^{\mathrm{T}}$ and interchanging $i$ and $j$ yields the result.

\end{proof}
Now considering (\ref{eq:backward_recursion}) and the approximation (\ref{eq:backward_appr_x}), define the probability mass functions:
\begin{equation*}
\mu_{s|t}(\mathbf{x}_s) \coloneqq \sum_{\mathbf{x}_{s+1}\in\mathcal{S}_{m,n}} \mu_{s+1|t}(\mathbf{x}_{s+1}) \frac{\mu_{s|s}(\mathbf{x}_{s})M_{s+1}(\mathbf{x}_{s},\bm{\pi}_{s|s},\mathbf{x}_{s+1})}{\sum_{\mathbf{z}\in\mathcal{S}_{m,n}}\mu_{s|s}(\mathbf{z})M_{s+1}(\mathbf{z},\bm{\pi}_{s|s},\mathbf{x}_{s+1})},\quad s<t,
\end{equation*}
recalling from (\ref{eq:backward_appr_x}) that  $\mu_{t|t}(\cdot)$ is the probability mass function associated with $\mathrm{Mult}(n,\bm{\pi}_{t|t})$. 
\begin{lemma}\label{lem:smoothing_recursion} For $0\leq s \leq t$,
	$\mu_{s|t}(\cdot)$ is the probability mass function associated with $\mathrm{Mult}(n,\bm{\pi}_{s|t})$, where $\bm{\pi}_{s|t}$ is computed as per algorithm \ref{alg:smoothing_x}.
\end{lemma}
\begin{proof}
	The result can be proved by induction. The induction is initialized using the fact that $\mu_{t|t}(\cdot)$ is by definition the probability mass function associated with $\mathrm{Mult}(n,\bm{\pi}_{t|t})$, and then proceeds by combining the result of lemma \ref{lem:smoothing_x_kernel} with moment generating function techniques similar to those used in the proof of lemma \ref{lem:prediction}. The details are omitted to avoid repetition.
\end{proof}

\subsection{Smoothing with the observation model from section  \ref{sec:obs_model_z}} \label{sec:smoothing_Z}
Assuming algorithm \ref{alg:filtering_inc} has already been run up to a given time $t$,
our objective in this section is to derive algorithm \ref{alg:smoothing_z}, from which we take the approximation:
\begin{equation*}
p(\mathbf{Z}_s|\mathbf{y}_{1:t})\approx \mathrm{Mult}(n,\mathbf{P}_{s|t}).
\end{equation*}

\begin{algorithm}
	\caption{Multinomial smoothing with observations derived from $(\mathbf{Z}_t)_{t\geq1}$}\label{alg:smoothing_z}
	\begin{algorithmic}[1]

		\For{$s=t-1,\ldots, 1$}
		\State $\bm{\pi}_{s|t} \leftarrow \mathbf{P}_{s+1|t} \mathbf{1}_m$
		\State let $\overline{\mathbf{L}}_s$ be the matrix with elements $\overline{l}_s^{(i,j)}\leftarrow p_{s|s}^{(j,i)}/ \pi_{s|s}^{(i)}$
		\State $\mathbf{P}_{s|t} \leftarrow (\mathbf{1}_m  \otimes \bm{\pi}_{s|t})\circ\overline{\mathbf{L}}_s^{\mathrm{T}}  $

		\EndFor
	\end{algorithmic}
\end{algorithm}
In  algorithm \ref{alg:smoothing_z}, $p_{s|s}^{(i,j)}$ are the elements of $\mathbf{P}_{s|s}$ and $\pi_{s|s}^{(i)}$ are the elements of $\bm{\pi}_{s|s}\coloneqq(\mathbf{1}_m^{\mathrm{T}}\mathbf{P}_{s|s})^{\mathrm{T}}$, with $\mathbf{P}_{s|s}$ computed in algorithm \ref{alg:filtering_inc}.

Similarly to section \ref{sec:smoothing_x},  to derive our approximations we start from the fact that $(\mathbf{Z}_t,\mathbf{Y}_t)_{t\geq1}$ is a hidden Markov model, and consider the identities:
\begin{align}
& p(\mathbf{Z}_{1:t}|\mathbf{Y}_{1:t})=p(\mathbf{Z}_{t}|\mathbf{Y}_{1:t})\prod_{s=1}^{t-1}p(\mathbf{Z}_{s}|\mathbf{Z}_{s+1},\mathbf{Y}_{1:s})\nonumber\\
& p(\mathbf{Z}_{s}|\mathbf{Z}_{s+1},\mathbf{Y}_{1:s})=\frac{ p(\mathbf{Z}_{s}|\mathbf{Y}_{1:s})p(\mathbf{Z}_{s+1}|\mathbf{Z}_{s})}{p(\mathbf{Z}_{s+1}|\mathbf{Y}_{1:s})}\label{eq:smoothing_kernel_Z}\\
&p(\mathbf{Z}_{s+1}|\mathbf{Y}_{1:s})=\sum_{\mathbf{Z}_{s}} p(\mathbf{Z}_{s}|\mathbf{Y}_{1:s})p(\mathbf{Z}_{s+1}|\mathbf{Z}_{s}).\nonumber
\end{align}
The backward recursion is in this case:
\begin{equation}
\sum_{\mathbf{Z}_{s+1}}p(\mathbf{Z}_{s+1}|\mathbf{Y}_{1:t}) p(\mathbf{Z}_{s}|\mathbf{Z}_{s+1},\mathbf{Y}_{1:s})= p(\mathbf{Z}_{s}|\mathbf{Y}_{1:t}). \label{eq:backward_recursion_z}
\end{equation}
Writing $\overline{\mu}_{s|s}(\cdot)$ for the probability mass function associated with $\mathrm{Mult}(n,\mathbf{P}_{s|s})$, our approximation to (\ref{eq:smoothing_kernel_Z}) is:
\begin{equation}
p(\mathbf{Z}_{s}|\mathbf{Z}_{s+1},\mathbf{Y}_{1:s}) \approx \frac{\overline{\mu}_{s|s} (\mathbf{Z}_s) \overline{M}_{s+1}(\mathbf{Z}_{s}, \bm{\pi}_{s|s},\mathbf{Z}_{s+1})}{\sum_{\widetilde{\mathbf{Z}}}  
	\overline{\mu}_{s|s} (\widetilde{\mathbf{Z}}) \overline{M}_{s+1}(\widetilde{\mathbf{Z}}, \bm{\pi}_{s|s},\mathbf{Z}_{s+1})}.\label{eq:kernel_approx_Z}
\end{equation}
where $\overline{M}_{s+1}$ was introduced in section \ref{sec:Z_approximations} and in the setting of (\ref{eq:kernel_approx_Z}) has the explicit formula: 
\begin{equation*}
\overline{M}_{s+1}(\mathbf{Z}_{s}, \bm{\pi}_{s|s},\mathbf{Z}_{s+1})=\mathbb{I}[\mathbf{1}_{m}^{\mathrm{T}}\mathbf{Z}_{s}=(\mathbf{Z}_{s+1}\mathbf{1}_{m})^{\mathrm{T}}]\left(\prod_{j=1}^{m}(\mathbf{Z}_{s+1}\mathbf{1}_{m})^{(j)}!\prod_{\ell=1}^{m}\frac{\left(k_{s+1,\bm{\pi}_{s|s}}^{(j,\ell)}\right)^{z_{s+1}^{(j,\ell)}}}{z_{s+1}^{(j,\ell)}!}\right).
\end{equation*}

\begin{lemma}\label{lem:smoothing_Z_kernel}
	With $\mathbf{Z}_{s+1}$ considered fixed, the function which maps $\mathbf{Z}_{s}$ to the right hand side of
	(\ref{eq:kernel_approx_Z}) is the probability mass function such that the columns of $\mathbf{Z}_{s}$ are independent and the distribution of the $i$th column is $\mathrm{Mult}((\mathbf{Z}_{s+1}\mathbf{1}_m)^{(i)},\overline{\mathbf{L}}_{s}^{(i,\cdot)})$, where
	$\overline{\mathbf{L}}_{s}$ is the row-stochastic matrix with entries
	\begin{equation*}
	\overline{l}_s^{(i,j)} =  p_{s|s}^{(j,i)}/ \pi_{s|s}^{(i)}
	\end{equation*}
	where $p_{s|s}^{(i,j)}$ are the elements of $\mathbf{P}_{s|s}$, and $\pi_{s|s}^{(i)}$ are the elements of $\bm{\pi}_{s|s}\coloneqq(\mathbf{1}_m^{\mathrm{T}}\mathbf{P}_{s|s})^{\mathrm{T}}$ with $\mathbf{P}_{s|s}$ computed in algorithm \ref{alg:filtering_inc}.
\end{lemma}
\begin{proof}
	For the numerator in (\ref{eq:kernel_approx_Z}) is
	\begin{align}
	&\overline{\mu}_{s|s} (\mathbf{Z}_s) \overline{M}_{s+1}(\mathbf{Z}_{s}, \bm{\pi}_{s|s},\mathbf{Z}_{s+1}) \nonumber\\
	&= \left(n! \prod_{i,j=1}^m \frac{(p_{s|s}^{(i,j)})^{z_s^{(i,j)}}} {z_s^{(i,j)}!}\right)\mathbb{I}[\mathbf{1}_{m}^{\mathrm{T}}\mathbf{Z}_{s}=(\mathbf{Z}_{s+1}\mathbf{1}_{m})^{\mathrm{T}}]\nonumber\\
	&\quad\times\left(\prod_{j=1}^{m}(\mathbf{Z}_{s+1}\mathbf{1}_{m})^{(j)}!\prod_{\ell=1}^{m}\frac{\left(k_{s+1,\bm{\pi}_{s|s}}^{(j,\ell)}\right)^{z_{s+1}^{(j,\ell)}}}{z_{s+1}^{(j,\ell)}!}\right),\label{eq:smooth_z_num}
	\end{align}
	and for the denominator in  (\ref{eq:kernel_approx_Z}),
	\begin{align}
	&\sum_{\mathbf{Z}_s}\overline{\mu}_{s|s} (\mathbf{Z}_s) \overline{M}_{s+1}(\mathbf{Z}_{s},\bm{\pi}_{s|s},\mathbf{Z}_{s+1})\nonumber \\
	&= \left(n! \prod_{j=1}^m \frac{(\pi_{s|s}^{(j)})^{(\mathbf{Z}_{s+1}\mathbf{1}_{m})^{(j)}}} {(\mathbf{Z}_{s+1}\mathbf{1}_{m})^{(j)}!}\right)\left(\prod_{j=1}^{m}(\mathbf{Z}_{s+1}\mathbf{1}_{m})^{(j)}!\prod_{\ell=1}^{m}\frac{\left(k_{s+1,\bm{\pi}_{s|s}}^{(j,\ell)}\right)^{z_{s+1}^{(j,\ell)}}}{z_{s+1}^{(j,\ell)}!}\right),\label{eq:smooth_z_den}
	\end{align}
	where the equality  in (\ref{eq:smooth_z_den}) holds by combining (\ref{eq:smooth_z_num}) with the fact that the marginal of $\overline{\mu}_{s|s}$ over $\mathbf{1}_{m}^{\mathrm{T}}\mathbf{Z}_{s}$ is $\mathrm{Mult}(n,\mathbf{1}_{m}^{\mathrm{T}}\mathbf{P}_{s|s})=\mathrm{Mult}(n,\bm{\pi}_{s|s})$.
	
	Dividing (\ref{eq:smooth_z_num}) by (\ref{eq:smooth_z_den}) results in
	\begin{equation}
	\mathbb{I}[\mathbf{1}_{m}^{\mathrm{T}}\mathbf{Z}_{s}=(\mathbf{Z}_{s+1}\mathbf{1}_{m})^{\mathrm{T}}]\prod_{j=1}^m (\mathbf{Z}_{s+1}\mathbf{1}_{m})^{(j)} !  \prod_{i=1}^m \left(\frac{p_{s|s}^{(i,j)}}{\pi_{s|s}^{(j)}}\right)^{z_s^{(i,j)}!} \frac{1}{z_s^{(i,j)}!},\label{eq:backward_kernel_approx_z}
	\end{equation}
	where $p_{s|s}^{(i,j)}$ are the elements of $\mathbf{P}_{s|s}$, $\overline{\mathbf{L}}_s$ is the matrix with elements $\overline{l}_{s}^{(i,j)}=p_{s|s}^{(j,i)}/\pi_{s|s}^{(i)}$.
\end{proof}

Now considering (\ref{eq:backward_recursion_z}) and the approximation (\ref{eq:kernel_approx_Z}), define the probability mass functions:
\begin{equation}
\overline{\mu}_{s|t}(\mathbf{Z}_s) \coloneqq \sum_{\mathbf{Z}_{s+1}} \overline{\mu}_{s+1|t}(\mathbf{Z}_{s+1}) \frac{\overline{\mu}_{s|s}(\mathbf{Z}_{s})\overline{M}_{s+1}(\mathbf{Z}_{s},\bm{\pi}_{s|s},\mathbf{Z}_{s+1})}{\sum_{\widetilde{\mathbf{Z}}}\overline{\mu}_{s|s}(\widetilde{\mathbf{Z}})\overline{M}_{s+1}(\widetilde{\mathbf{Z}},\bm{\pi}_{s|s},\mathbf{Z}_{s+1})},\quad s<t,\label{eq:mu_bar_defn}
\end{equation}
where $\overline{\mu}_{t|t}(\cdot)$ is defined to be the probability mass function associated with $\mathrm{Mult}(n,\mathbf{P}_{t|t})$. 
\begin{lemma} For $0\leq s \leq t$,
	$\overline{\mu}_{s|t}(\cdot)$ is the probability mass function associated with $\mathrm{Mult}(n,\mathbf{P}_{s|t})$, where $\mathbf{P}_{s|t}$ is computed as per algorithm \ref{alg:smoothing_z}.
\end{lemma}
\begin{proof} The proof is by induction for $s=t,t-1,\ldots$, initialized using the fact that $\overline{\mu}_{t|t}(\cdot)$  is defined to be the probability mass function associated with $\mathrm{Mult}(n,\mathbf{P}_{t|t})$, and for the induction step plugging (\ref{eq:backward_kernel_approx_z}), which is an explicit expression for the right hand side of (\ref{eq:kernel_approx_Z}), into (\ref{eq:mu_bar_defn}), and using the fact the marginal of $\overline{\mu}_{s+1|t}(\mathbf{Z}_{s+1})$ over $\mathbf{Z}_{s+1} \mathbf{1}_m$ is $\mathrm{Mult}(n,\mathbf{P}_{s+1|t}\mathbf{1}_m)$.
\end{proof}

\section{Ebola example: further details and numerical results}\label{sec:ebola_supp}
In this section we provide more information about the numerical results from section \ref{sec:ebola} in the main part of the paper.


\subsection{Model}
The model from \cite{lekone2006statistical} is a discrete-time stochastic SEIR model with time varying $\beta$. $\mathbf{K}_{t,\bm{\eta}}$ is as in (\ref{eq:SEIR_as_ind}) with $h=1$ and with $\beta$ replaced by
\begin{equation}
\beta_{t} = 
\begin{cases}
\beta                     & t< t_\star\\
\beta e^{-\lambda(t-t_\star)} & t \geq t_\star
\end{cases},
\end{equation}
where $t_\star$ has the interpretation as the day on which control measures were introduced.
Following \cite{lekone2006statistical}, the initial distribution was fixed to $\bm{\pi}_0 = [1-1/n,1/n,0,0]^\mathrm{T}$ with $n=5,364,501$. The observation model from section \ref{sec:obs_model_z} was used with $q_t^{(i,j)}=0$ for all $t$ and all $(i,j)$ except $(2,3)$ and $(3,4)$, and where $q_t^{(2,3)}\equiv q^{(2,3)}$ and $q_t^{(3,4)}\equiv q^{(3,4)}$ are treated as constant-in-$t$ and to be estimated.
The parameters of the model are thus:
\begin{equation*}
\Theta = (\beta, \lambda, \rho, \gamma, q^{(2,3)}, q^{(3,4)}).
\end{equation*}

\subsection{Details of the EM algorithm}
In numerical experiments we found that a robust approach to approximate maximum likelihood estimation of $\Theta$ was to take a profile-likelihood approach using an EM algorithm: 1) choose a grid of values for $(\beta,\lambda)$; 2) for each point on this grid, say $(\hat{\beta},\hat{\lambda})$, run an EM algorithm to approximately maximize $p(\mathbf{Y}_{1:t}|\hat{\beta},\hat{\lambda},\rho,\gamma,q^{(2,3)}, q^{(3,4)})$ with respect to $(\rho,\gamma,q^{(2,3)}, q^{(3,4)})$ then evaluate the marginal likelihood at the resulting parameter values using algorithm \ref{alg:filtering_inc}; 3) maximize over the grid.

The EM component of this procedure follows the usual steps for a hidden Markov model \citep{cappe2006inference}, so we just provide an outline.  One step of the EM procedure is as follows: given $\Theta$ one performs forward filtering using algorithm \ref{alg:filtering_inc} then backward smoothing using algorithm \ref{alg:smoothing_z} resulting in $(\mathbf{P}_{s|t})_{s\leq t}$. The expected complete data log-likelihood is then maximized with respect to the parameters of interest. It turns out that for the Ebola model, the maximization steps for $(\rho,\gamma,q^{(2,3)}, q^{(3,4)})$ have closed-form solutions, leading to the update equations:
\begin{align*}
\rho \leftarrow \log\left(1+\frac{\sum_{s=1}^t p_{s|t}^{(2,3)}}{\sum_{s=1}^t p_{s|t}^{(2,2)}}\right),
\qquad
\gamma  \leftarrow \log\left(1+\frac{\sum_{s=1}^t p_{s|t}^{(3,4)}}{\sum_{s=1}^t p_{s|t}^{(3,4)}}\right), \\
q^{(2,3)} \leftarrow 1 \wedge \frac{\sum_{s=1}^t y_s^{(2,3)}/n}{p_{s|t}^{(2,3)}},
\qquad
q^{(3,4)} \leftarrow 1 \wedge \frac{\sum_{s=1}^t y_s^{(3,4)}/n}{p_{s|t}^{(3,4)}},
\end{align*}
where $p_{s|t}^{(i,j)}$ are the elements of $\mathbf{P}_{s|t}$.

\subsection{Details of the MCMC algorithm}
We  implemented a Metropolis-within-Gibbs MCMC algorithm targeting the approximate posterior distribution $\widehat{p}(\Theta|\mathbf{Y}_{1:t})\propto \widehat{p}(\mathbf{Y}_{1:t}|\Theta)p(\Theta)$, where $\widehat{p}(\mathbf{Y}_{1:t}|\Theta)$ is the approximate marginal likelihood computed using algorithm \ref{alg:filtering_inc}, and 
\begin{equation*}
p(\Theta) = p(\beta)p(\lambda)p(\rho)p(\gamma)p(q^{(2,3)})p(q^{(3,4)}).
\end{equation*}
We considered the three sets of Gamma prior distributions over $\beta,\lambda,\rho,\gamma$ specified in section 3.3 of \cite{lekone2006statistical} and referred to as `vague', `informative' and `non-centered'. The priors $p(q^{(2,3)}),p(q^{(3,4)}))$ were taken to be uniform densities on $[0,1]$.

Gaussian random walk proposals were applied to each parameter, with variances manually tuned to give acceptance rates between $20\%$ and $40\%$ \citep{roberts2001optimal}.

\subsection{Synthetic data: supplementary plots for our MCMC method}
To generate the data we used the following parameter values from \cite{lekone2006statistical}:
\begin{equation}
\beta = 0.2, \quad \lambda = 0.2, \quad \rho = 0.2, \quad \gamma = 0.143,\quad t_\star=130,\label{eq:true_params}
\end{equation}
together with  $q^{(2,3)}=291/316$ and $q^{(3,4)}=236/316$. The data are shown in figure \ref{fig:part_obs_epidemic}.
\begin{figure}[httb!]
	\centering
	\resizebox{\columnwidth}{!}{
		\includegraphics[scale = 1]{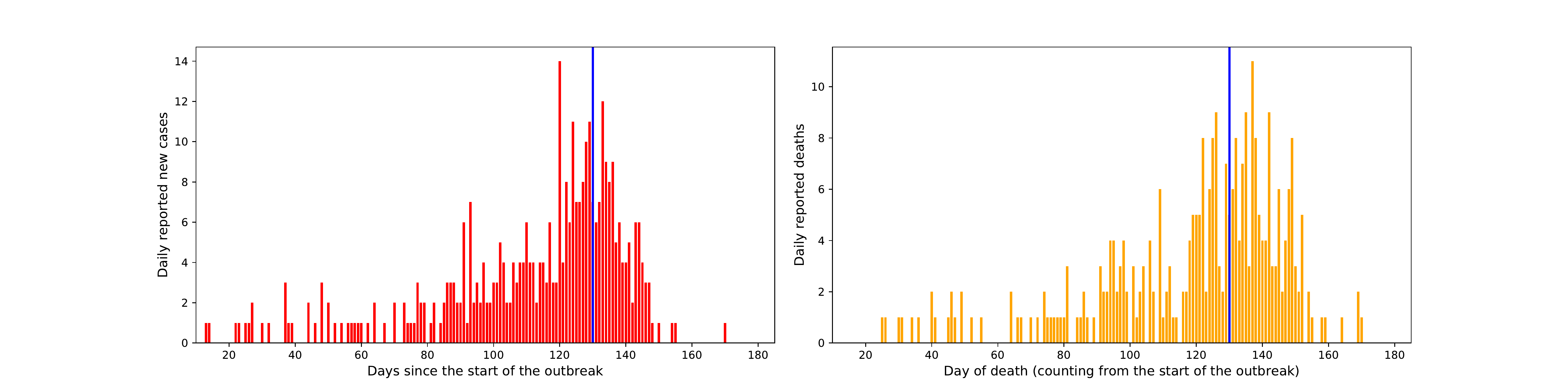}
	}
	\caption{Synthetic data $(y_t^{(2,3)})_{t\geq 1}$, which are daily numbers of reported new cases, and $(y_t^{(3,4)})_{t\geq 1}$, which are the daily numbers of reported new deaths, simulated from the Ebola model of  \cite{lekone2006statistical}. Blue lines indicate the day $t_\star=130$ at which control measures were introduced. } \label{fig:part_obs_epidemic}
\end{figure}
Figures \ref{fig:traceplot_q_less_1_learnt-Vague} and \ref{fig:traceplot_q_less_1_learnt-noncentered} show traceplots and histograms of the MCMC output from which the point estimates and posterior standard deviations in table \ref{tab:ebola_synthetic} were calculated. The MCMC chain  was run for $5\times10^5$ iterations, the first $10^5$ iterations were discarded for burn-in, and the remaining samples thinned to result in a sample size of $10^4$.
\begin{figure}[httb!]
	\centering
	\resizebox{\columnwidth}{!}{
		\includegraphics[]{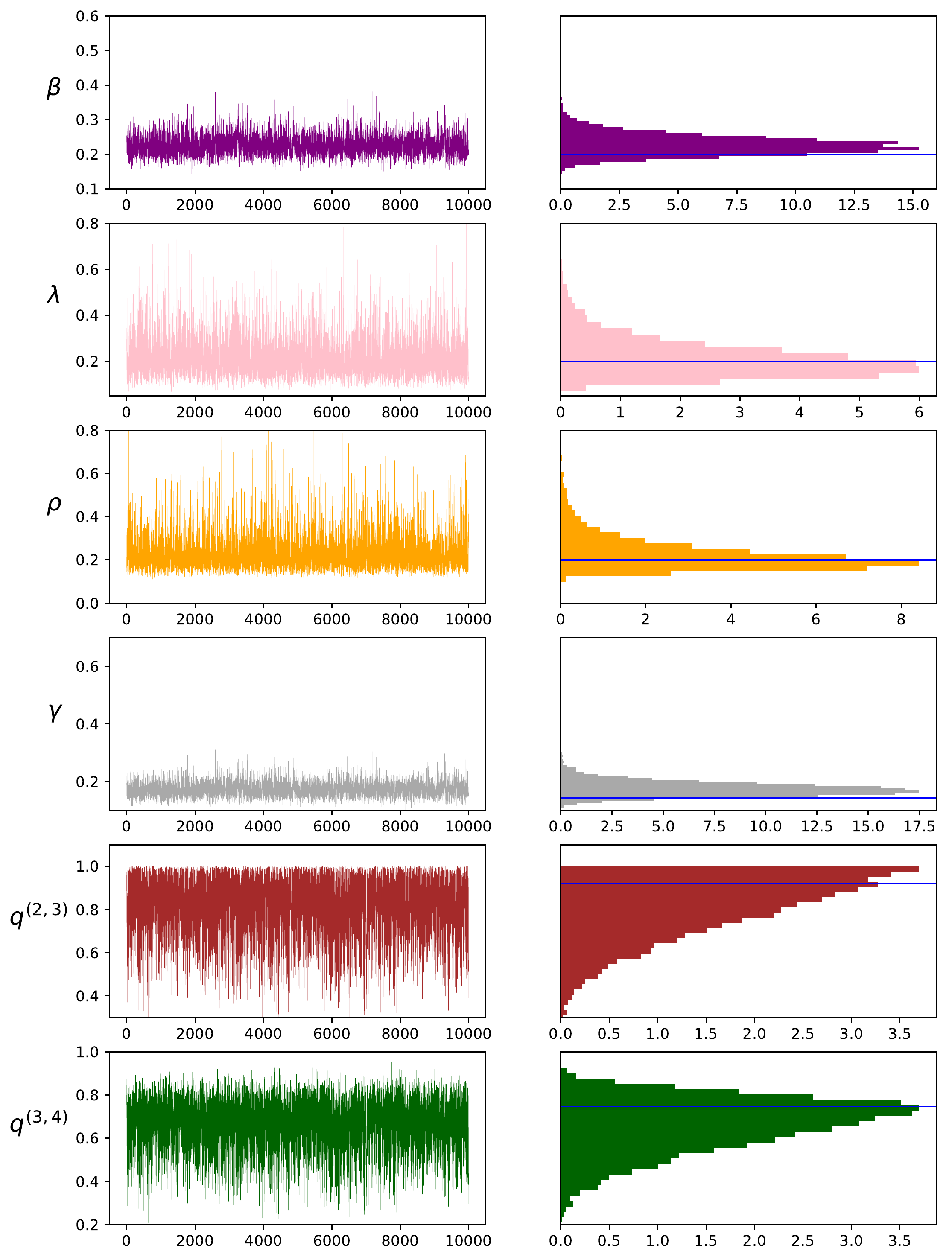}
		\includegraphics[]{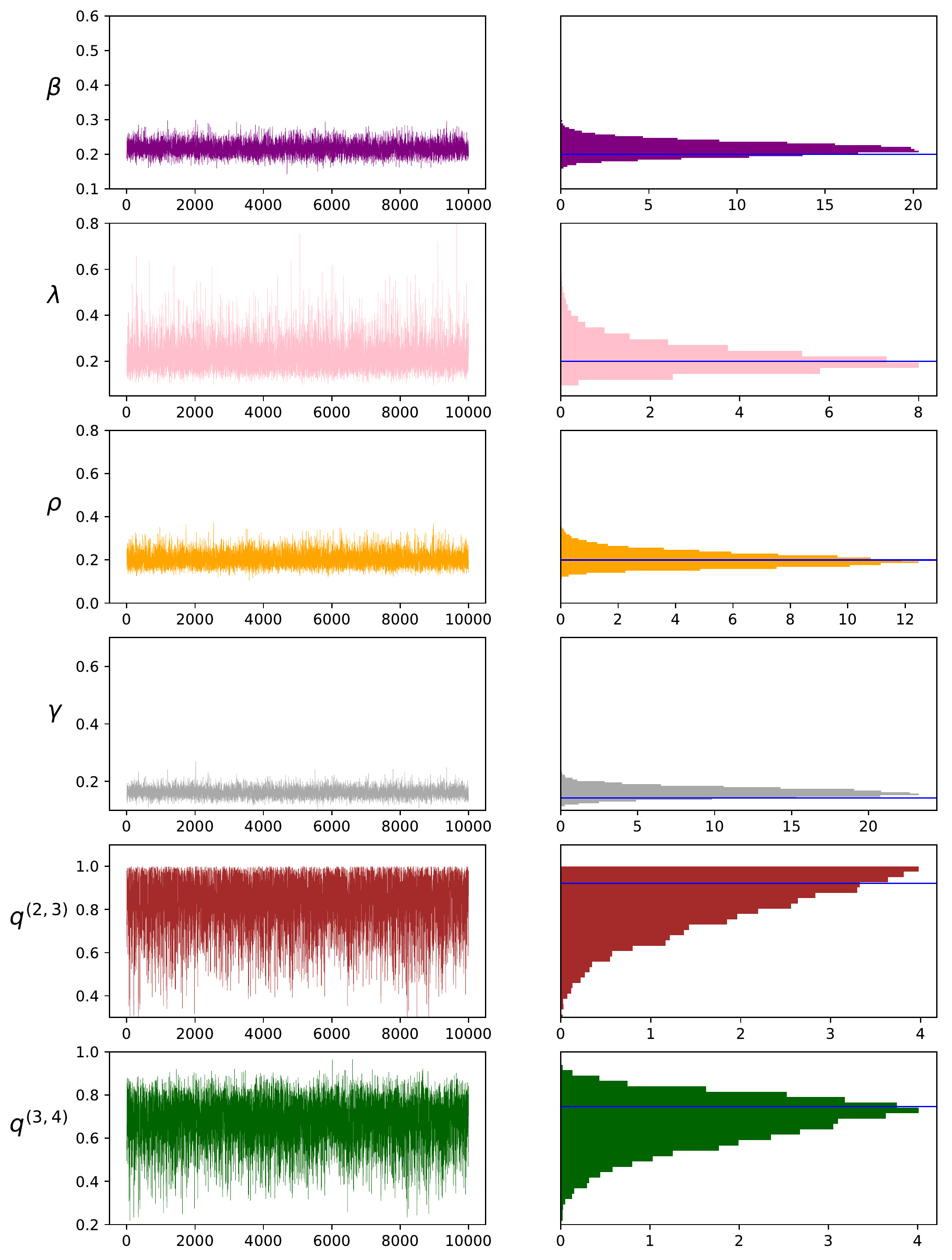}
	}
	\caption{Traceplots and histograms for our MCMC method applied to synthetic data from the Ebola model, with the `vague' set of priors (left) and the `informative' set of priors (right) specified by  \cite{lekone2006statistical}. Blue lines show true parameter values. } \label{fig:traceplot_q_less_1_learnt-Vague}
\end{figure}
\begin{figure}[httb!]
	\centering
	\resizebox{0.45\columnwidth}{!}{
		\includegraphics[]{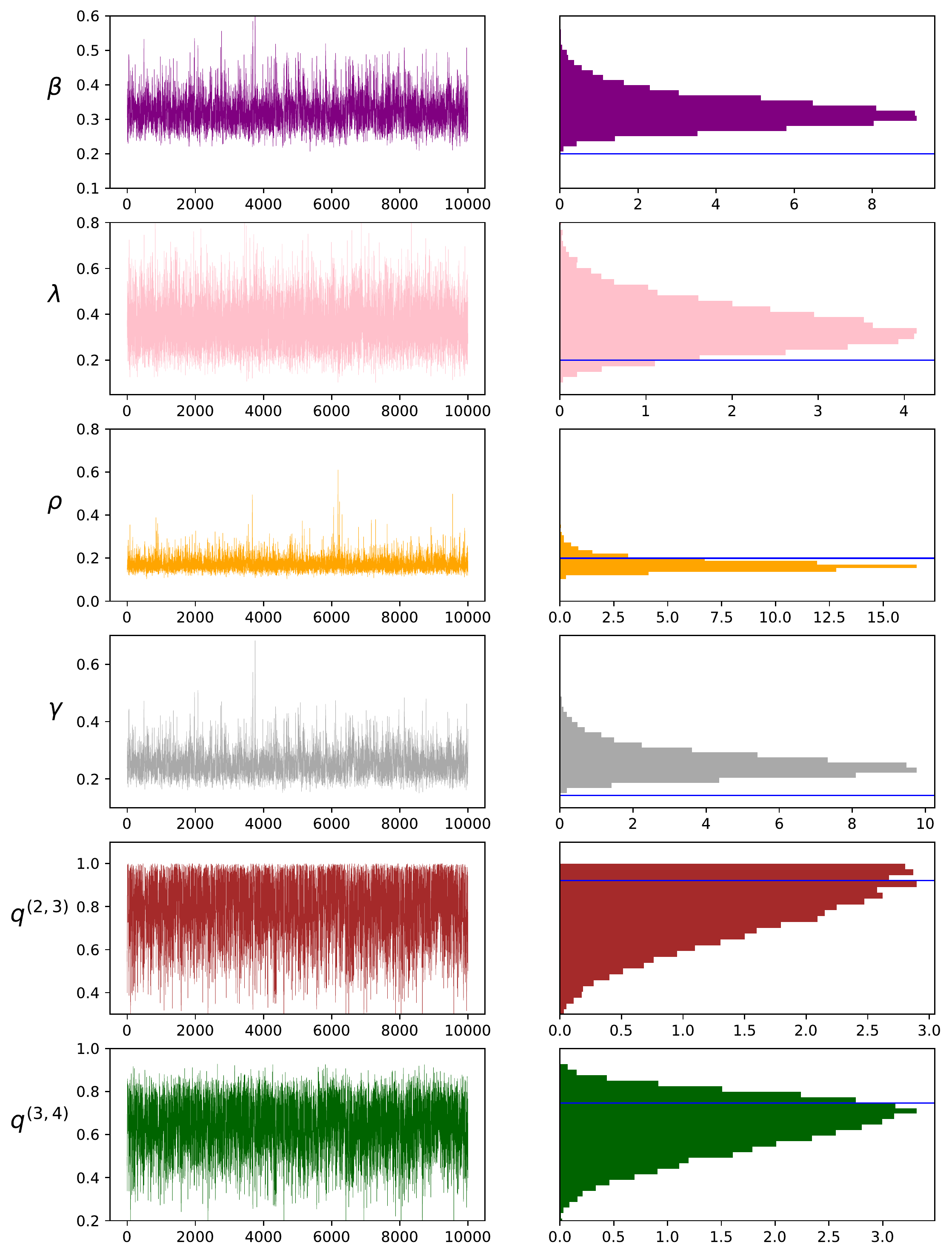}
	}
	\caption{Traceplots and histograms for our MCMC method applied to synthetic data from the Ebola model, with the `noncentered' set of priors specified by  \cite{lekone2006statistical}. Blue lines show true parameter values. } \label{fig:traceplot_q_less_1_learnt-noncentered}
\end{figure}

\newpage
\subsection{Real data: supplementary plots and further details for our MCMC method}



The data are shown in figure \ref{fig:Cases-ebola}. Traceplots and histograms for our MCMC method are displayed in figure \ref{fig:traceplot_hist_ebola}, we considered only the `vague' and `uninformative' sets of priors. 

\begin{figure}[httb!]
	\centering
	\resizebox{\columnwidth}{!}{
		\includegraphics[scale = 1]{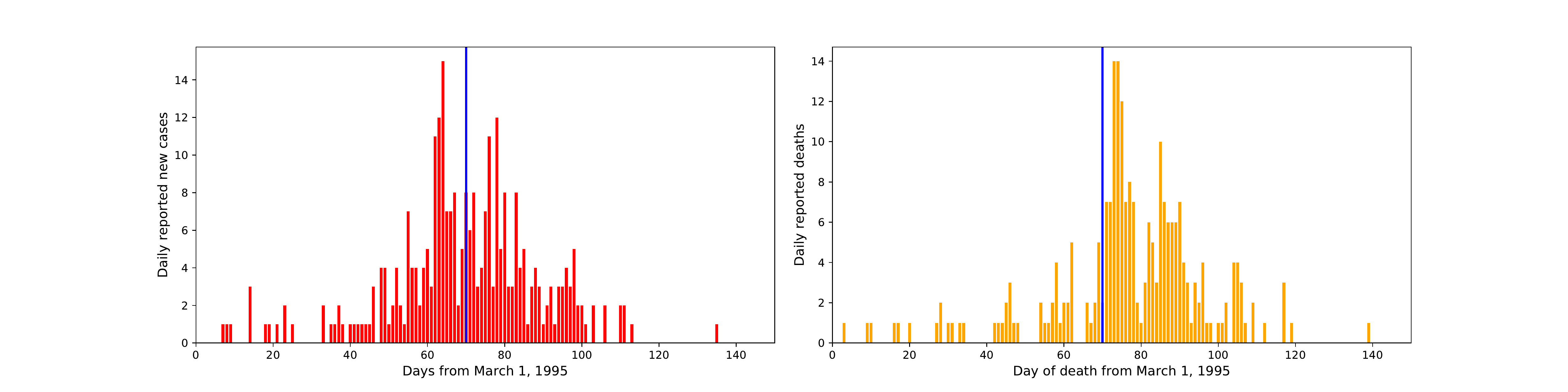}
	}
	\caption{Data from the 1995 Ebola outbreak in the Democratic Republic of Congo per day from March 1, 1995 to July 16, 1995. Blue lines indicate the 9th of May when control measurements were introduced.} \label{fig:Cases-ebola}
\end{figure}

It can be observed that the posteriors over $\beta, \lambda, \rho$ appear to be bimodal under the `vague' prior. According to this observation, we separated the MCMC samples out  by applying the $k$-means algorithm to the marginal samples for $\rho$. Qualitatively, this identifies:
\begin{enumerate}
	\item a mode with big $\beta$, big $\lambda$ and small $\rho$; 
	\item a mode with small $\beta$, small $\lambda$ and big $\rho$.
\end{enumerate}

\begin{figure}[httb!]
	\centering
	\resizebox{0.8\columnwidth}{!}{
		\includegraphics[scale = 1]{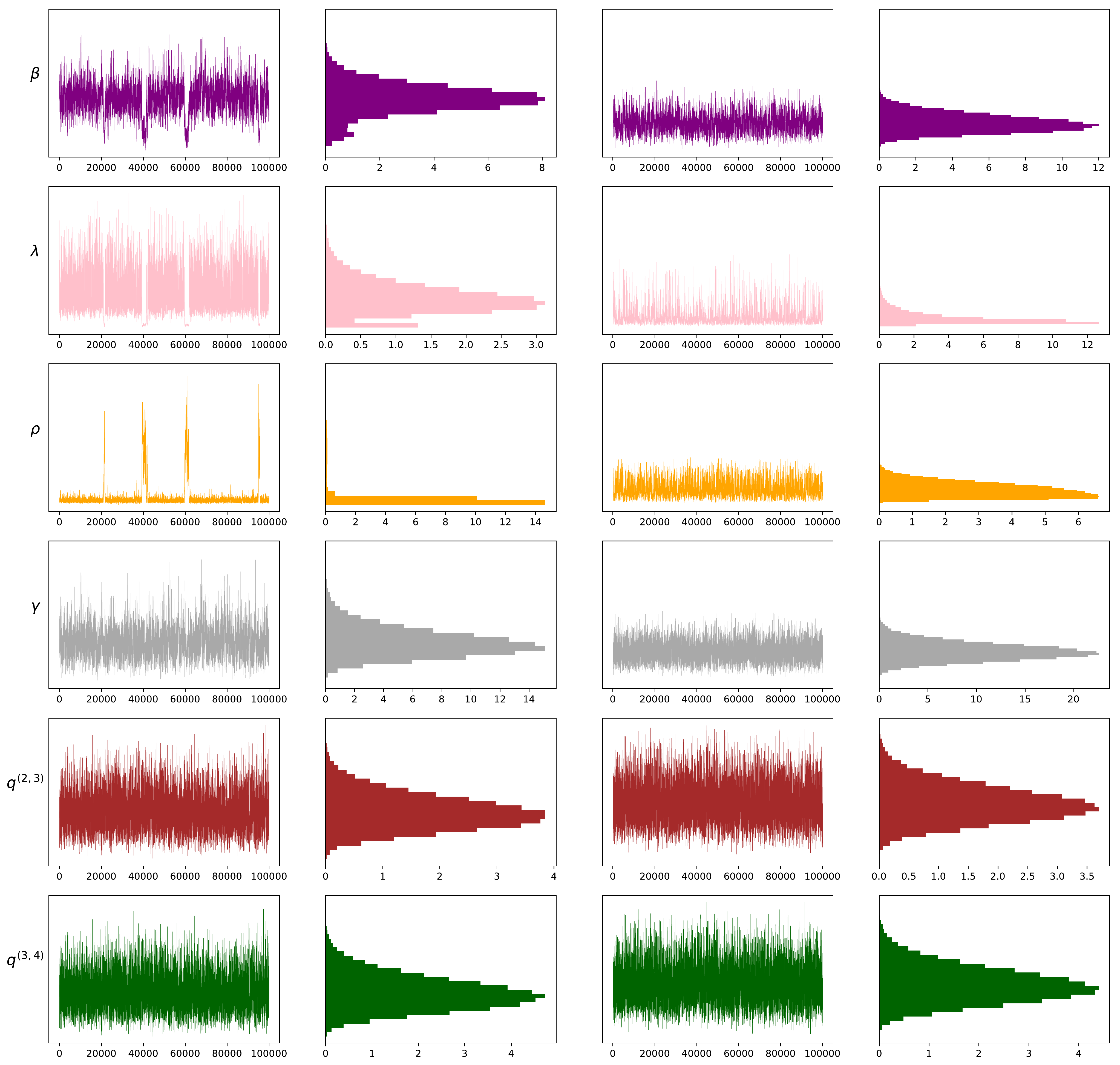}
	}
	\caption{Traceplots and histogram for the MCMC on the Ebola data.} \label{fig:traceplot_hist_ebola}
\end{figure}

This can be interpreted in terms of two alternative explanations of the observed data: the first one consists of a big initial growth of the epidemic (big $\beta$) with a slow transition to the infective state (small $\rho$), followed by an effective intervention (big $\lambda$) that slows down the spread of the virus; the second one is essentially the opposite, a small initial growth of the epidemic (small $\beta$) with a fast transition to the infective state (big $\rho$), followed by mild control measures (small $\lambda$). Posterior means and standard deviations are reported in Table \ref{tab:ebola}, along with the values for the Data Augmentation MCMC method of \cite{chowell2004basic}, and the least-squares method of \cite{lekone2006statistical}. For the ABSEIR ABC method of \cite{brown2016empirically} we obtained the estimates ourselves by running the code available at \url{http://grantbrown.github.io/ABSEIR/vignettes/Kikwit.html}.

Observe that the estimate for $\beta$ from \cite{lekone2006statistical} is close to our first mode while that from \cite{chowell2004basic} is closer to our second mode. We conjecture that the analyses in these works were each exploring only a single combination of parameters and excluding the other, and  we conjecture that our MCMC method mixes more quickly than that of \cite{lekone2006statistical}.  

\section{COVID-19 example: further details and numerical results}\label{sec:covid_sup}

\subsection{Model}
We consider the following discrete-time, stochastic version of the ODE model from  \cite{kucharski2020early}.  
\begin{itemize}
	\item $S_t$ is the number of susceptible individuals in Wuhan (compartment 1); 
	\item $E^{(1,W)}_{t},E^{(2,W)}_{t}$ are the numbers of exposed individuals in Wuhan in the first and second stage of the incubation period (compartments 2\&3); 
	\item $I^{(1,W)}_{t}, I^{(2,W)}_{t}$ are the numbers of infective individuals in Wuhan in the first and second stage of the disease (compartments 3\&4);
	\item $E^{(1,T)}_{t},E^{(2,T)}_{t}$ are the numbers of individuals who were initially exposed whilst in Wuhan and subsequently travelled to other countries, in their first and second stage of the incubation period (compartments 5\&6); 
	
	\item $I^{(1,T)}_{t}, I^{(2,T)}_{t}$ are the numbers of infective individuals who were initially exposed whilst in Wuhan and subsequently travelled to other countries, in their first and second stage of the disease (compartments 7\&8);
	\item $R_{t}$ is the number of removed individuals (compartment 9). 
\end{itemize}
The evolution of the compartments is as follows:
\begin{align}
&S_{t+1}       = S_{t}       - B^{(1,W)}_{t}       - B^{(1,T)}_{t},     
\label{eq:compartments_evolutions_1}\\
&E^{(1,W)}_{t+1} = E^{(1,W)}_{t} + B^{(1,W)}_{t}       - B^{(2,W)}_{t}, \\
&E^{(2,W)}_{t+1} = E^{(2,W)}_{t} + B^{(2,W)}_{t} - C^{(1,W)}_{t}, \\
&I^{(1,W)}_{t+1} = I^{(1,W)}_{t} + C^{(1,W)}_{t} - C^{(2,W)}_{t}, \\
&I^{(2,W)}_{t+1} = I^{(2,W)}_{t} + C^{(2,W)}_{t} - D^{W}_{t},     \\
&E^{(1,T)}_{t+1} = E^{(1,T)}_{t} + B^{(1,T)}_{t}       - B^{(2,T)}_{t}, \\
&E^{(2,T)}_{t+1} = E^{(2,T)}_{t} + B^{(2,T)}_{t} - C^{(1,T)}_{t}, \\
&I^{(1,T)}_{t+1} = I^{(1,T)}_{t} + C^{(1,T)}_{t} - C^{(2,T)}_{t}, \\
&I^{(2,T)}_{t+1} = I^{(2,T)}_{t} + C^{(2,T)}_{t} - D^{T}_{t},     \\
\label{eq:compartments_evolutions_10}
&R_{t+1}       = R_{t}       + D^{(W)}_{t}       + D^{(T)}_{t},
\end{align}
where:
\begin{equation*}
\left[\begin{array}{c}
B_{t}^{(1,W)}\\
B_{t}^{(1,T)}\\
S_{t}-B_{t}^{(1,W)}-B_{t}^{(1,T)}
\end{array}\right]\sim\mathrm{Mult}\left(S_{t},\left[\begin{array}{c}
(1-f_{t})p_{t}\\
f_{t}p_{t}\\
1-p_{t}
\end{array}\right]\right),
\end{equation*}
\begin{align}
&B^{(2,W)}_{t}  \sim \mathrm{Bin} \left ( E^{(1,W)}_{t}, p_C \right ), \quad B^{(2,T)}_{t}  \sim \mathrm{Bin} \left ( E^{(1,T)}_{t}, p_C \right ), \\	
&C^{(1,W)}_{t} \sim \mathrm{Bin} \left ( E^{(2,W)}_{t}, p_C \right ), \quad C^{(1,T)}_{t} \sim \mathrm{Bin} \left ( E^{(2,T)}_{t}, p_C \right ), \\
&C^{(2,W)}_{t}  \sim \mathrm{Bin} \left ( I^{(1,W)}_{t}, p_R \right ), \quad C^{(2,T)}_{t}  \sim \mathrm{Bin} \left ( I^{(1,T)}_{t}, p_R \right ),\\
&D^{(W)}_{t}        \sim \mathrm{Bin} \left ( I^{(2,W)}_{t}, p_R \right ), \quad D^{(T)}_{t}        \sim \mathrm{Bin} \left ( I^{(2,T)}_{t}, p_R \right ),\label{eq:expanded_compartments_evolutions_10}
\end{align}
with
\begin{equation}
p_{t} = 1 - e^{ - h \beta_{t}({I^{(1,W)}_{t}+I^{(2,W)}_{t}}) \slash {n} }, \quad
p_C  = 1 -e^{-h 2\rho}, \quad
p_R  = 1 - e^{-h 2\gamma},
\end{equation}
where $f_t$ is the fraction of cases that depart from Wuhan to other countries at time $t$.  The ODE model of \cite{kucharski2020early} incorporates a a number of other compartments which are used to accumulate the numbers of individuals which have passed through certain states, but which otherwise do not play an active role in the model, hence we do not specify them here.

The time-varying transmission rate $(\beta_t)_{t\geq 0}$  follows a log-normal random walk:
\begin{equation*}
\beta_{t+1} = \beta_t \exp(V_t),\qquad V_t \sim\mathcal{N}(0,\sigma_V^2).
\end{equation*} 

The observations consist of the new infectives in Wuhan, $y_{t}^{(3,4)}$, and internationally, $y_{t}^{(7,8)}$, at each time step, subject to random under-reporting:
\begin{equation}\label{eq:Y_covid}
y_{t}^{(3,4)} \sim \mathrm{Bin} \left ( C_{t}^{(1,W)}, q^{(W)} \right ), \quad
y_{t}^{(7,8)} \sim \mathrm{Bin} \left ( C_{t}^{(1,T)}, q^{(T)} \right ). 
\end{equation}

\subsection{Data and Parameter settings}
The series $(y_{t}^{(3,4)})_{t\geq0}$ and $(y_{t}^{(7,8)})_{t\geq0}$, i.e. the reported numbers of new infectives in Wuhan and internationally, constitute two of the three data sets considered for inference by \cite{kucharski2020early}. They additionally considered a third data set consisting of information about prevalence of infections on evacuation flights. We exclude this prevalence data from our analysis, since the structure of the observation model required fall outside the class of models we consider in this paper.

We set $f_t, n, \rho, \gamma$ to the same values used in \cite{kucharski2020early}, including the fact that $f_t$  is set to zero after the date when travel restrictions were introduced. We estimated $q^{(W)}=0.00175,q^{(T)}=0.8$ via approximate maximum likelihood over a grid. We set $h=1$.

\subsection{Implementation}
We based our implementation directly on the R code accompanying \cite{kucharski2020early}, which is available at \url{https://github.com/ adamkucharski/2020-ncov/}.

We also re-ran the experiments reported in \cite{kucharski2020early} using their method, but excluding the evacuation flight data mentioned above. This allows for like-for-like comparisons of our results with theirs - see below.

\subsection{Inference}

\begin{algorithm}[h]
	\caption{Particle filter and backward sampler for COVID-19 application}\label{alg:particle_filter_new}
	\begin{algorithmic}[1]
		\State \textbf{initialize} $\bm{\pi}^{(i)}_{0|0}\leftarrow \bm{\pi}_0$, \quad $\beta_0^{(i)} = \beta_0,  \quad \text{for } i = 1,\dots, n_{\mathrm{part}}$
		\For{$s= 1$ to $t$}
		\For{$i= 1$ to $n_{\mathrm{part}}$}
		\vspace{0.1cm}
		\State{ $\beta_s^{(i)} \leftarrow \beta_{s-1}^{(i)} \exp(V^{(i)}), \quad V^{(i)} \sim \mathrm{N}(0,\sigma_V^2)$}
		\State $\mathbf{P}_{s|s-1}^{(i)} \leftarrow (\bm{\pi}^{(i)}_{s-1|s-1} \otimes \mathbf{1}_m ) \circ \mathbf{K}_{t,\bm{\pi}^{(i)}_{s-1|s-1}, \beta_{t}^{(i)}}$
		\State $\mathbf{P}^{(i)}_{s|s} \leftarrow  \dfrac{\mathbf{Y}_{s}}{n}+\left(1-\dfrac{\mathbf{1}_m^{\mathrm{T}}\mathbf{Y}_{s}\mathbf{1}_m}{n}\right)\dfrac{\mathbf{P}^{(i)}_{s|s-1}\circ(\mathbf{1}_m\otimes \mathbf{1}_m-\mathbf{Q}_{s})}{1 - \mathbf{1}_m^{\mathrm{T}}(\mathbf{P}^{(i)}_{s|s-1}\circ \mathbf{Q}_{s})\mathbf{1}_m}$
		\State $\log w_s^{(i)} \leftarrow \log(n!) + \mathbf{1}_m^{\mathrm{T}} (\mathbf{Y}_{s} \circ \log \mathbf{P}^{(i)}_{s|s-1} )\mathbf{1}_m + \mathbf{1}_m^{\mathrm{T}} (\mathbf{Y}_{s} \circ \log \mathbf{Q}_{s} )\mathbf{1}_m -  \mathbf{1}_m^{\mathrm{T}} \log(\mathbf{Y}_{s} !)\mathbf{1}_m$ \\
		\qquad\qquad\quad\;\; $+ (n-\mathbf{1}_m^{\mathrm{T}} \mathbf{Y}_{s} \mathbf{1}_m) \log (1 -\mathbf{1}_m^{\mathrm{T}} (\mathbf{P}^{(i)}_{s|s-1} \circ \mathbf{Q}_{s}) \mathbf{1}_m ) - \log ((n-\mathbf{1}_m^{\mathrm{T}} \mathbf{Y}_{s} \mathbf{1}_m)!)$
		\State $\bm{\pi}^{(p)}_{s|s}\leftarrow(\mathbf{1}_m^{\mathrm{T}}\mathbf{P}^{(i)}_{s|s})^{\mathrm{T}}$
		\EndFor
		\State $\bar{w}_{s}^{(i)} \leftarrow w^{(s)}_{i} \slash \sum_{j} w^{(j)}_{s}, \quad i =i,\dots, n_{\mathrm{part}} $
		\State \textbf{resample } $\{\beta_{s}^{(i)}, \bm{\pi}^{(i)}_{s|s}\}_{i=1}^{n_{\mathrm{part}}}$ according to $\{\bar{w}^{(i)}_{s}\}_{i=1}^{n_{\mathrm{part}}}$ and keep track of ancestors in $\bm{a}_{s}=[a_s^{(1)}\,\cdots\,a_s^{(n_{\mathrm{parts}})}]^{\mathrm{T}}$
		\EndFor
		\State
		\State \textbf{sample } $\zeta$	according to $\{\bar{w}^{(i)}_{t}\}_{i=1}^{n_{\mathrm{part}}}$
		\State $\tilde{\bm{\pi}}_{{t}|{t}} \leftarrow \bm{\pi}_{{t}|{t}}^{(\zeta)}$, \quad  $\tilde{\mathbf{P}}_{{t}|{t}} \leftarrow \mathbf{P}_{{t}|{t}}^{(\zeta)}$, \quad $\tilde{\beta}_t \leftarrow \beta_t^{(\zeta)}$
		\State \textbf{sample } $\tilde{\mathbf{Z}}_{t}$ from $\mathrm{Mult}(n,\tilde{\mathbf{P}}_{{t}|{t}})$
		
		\For{$s=t-1,\ldots, 1$}
		
		\State $\tilde{\bm{\pi}}_{s|{t}} \leftarrow \tilde{\mathbf{P}}_{s+1|{t}} \mathbf{1}_m$
		\State $\zeta \leftarrow a_s^{(\zeta)}$, \quad  $\tilde{\bm{\pi}}_{s|s} \leftarrow \bm{\pi}_{s|s}^{(\zeta)}$, \quad $\tilde{\mathbf{P}}_{s|s} \leftarrow \mathbf{P}_{s|s}^{(\zeta)}$, \quad $\tilde{\beta}_s \leftarrow \beta_s^{(\zeta)}$
		\State Let $\overline{\mathbf{L}}_s$ be the matrix with elements $\overline{l}_s^{(i,j)}\leftarrow \tilde{p}_{s|s}^{(j,i)}/ \tilde{\pi}_{s|s}^{(i)}$
		\For{$i=1,\ldots,m$}
		
		\State \textbf{sample } $\tilde{\mathbf{Z}}^{(\cdot, i)}_{s}$ from $\mathrm{Mult} ((\tilde{\mathbf{Z}}_{s+1}\mathbf{1}_{m})^{(i)}, \overline{\mathbf{L}}_{s}^{(i,\cdot)})$
		
		\EndFor
		\State $\tilde{\mathbf{P}}_{s|t} \leftarrow (\mathbf{1}_m  \otimes \tilde{\bm{\pi}}_{s|t})\circ\overline{\mathbf{L}}_s^{\mathrm{T}}  $

		\EndFor
		
		\Return $\{\tilde{\beta}_s, \tilde{\mathbf{P}}_{s|t}, \tilde{\mathbf{Z}}_s\}_{s=1}^t$
	\end{algorithmic}
\end{algorithm} 

Algorithm \ref{alg:particle_filter_new} is based directly on the Sequential Monte Carlo algorithm of \cite{kucharski2020early}, incorporates our discrete-time stochastic model instead of their ODE model. The first stage consists of a particle filter where for time step $s$, the $i$th of $n_{\mathrm{part}}$ particles consists of $\beta_s^{(i)}$ and  $\mathbf{P}_{s|s}^{(i)}$, and the unnormalized importance weight $w_s^{(i)}$ is computed similarly to algorithm \ref{alg:filtering_inc}. The second stage samples from the smoothing distribution of $\beta_{s}$ by tracing back the ancestors of a selected particle - see \citep{andrieu2010particle} for details of the role of ancestors in resampling. In addition backward steps as in algorithm \ref{alg:smoothing_z} compute the corresponding smoothing distribution over $\mathbf{Z}_{s}$, for $s=t,\ldots,1$. 


\subsection{Results}

The  effective sample size for algorithm \ref{alg:particle_filter_new} and the SMC method of \cite{kucharski2020early} applied to the same data are reported in figure \ref{fig:covid_ESS}. 

\begin{figure}[httb!]
	\begin{subfigure}{.45\textwidth}
		\includegraphics[scale = 0.7]{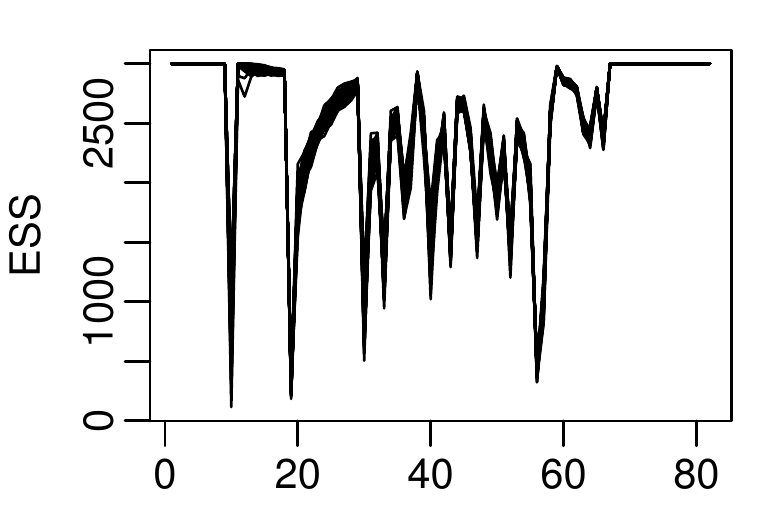}
	\end{subfigure}
	\hspace{0.1\textwidth}
	\begin{subfigure}{.45\textwidth}
		\includegraphics[scale = 0.7]{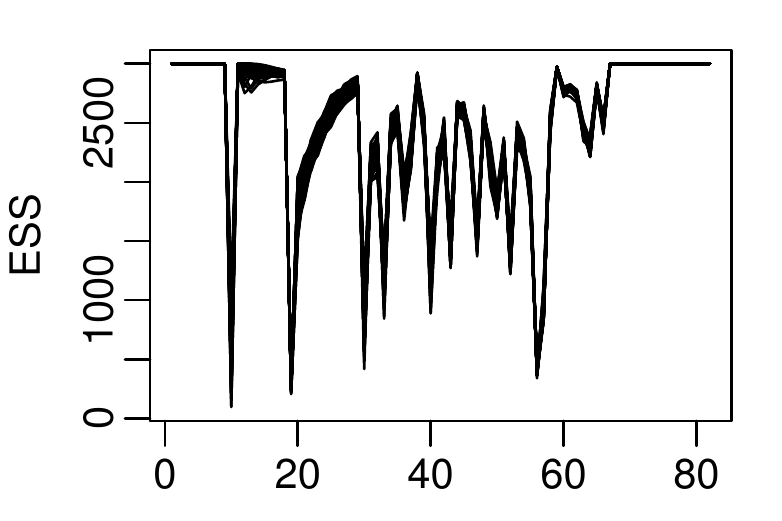}
	\end{subfigure}
	\caption{Effective sample size for algorithm \ref{alg:particle_filter_new} on the left, and for the the particle filter of \cite{kucharski2020early} on the right. Remark that we exclude the evacuation flights data from \cite{kucharski2020early} from the latter to make a fair comparison with our model.} \label{fig:covid_ESS}
\end{figure}

We followed the procedure of \cite{kucharski2020early} in producing the remainder of our results: we ran the algorithm $100$ times with $n_{\mathrm{part}}=3\times10^3$, resulting in 100 samples of $\{\tilde{\beta}_{s}, \tilde{\mathbf{P}}_{s|t}, \tilde{\mathbf{Z}}_{s}\}_{s=1}^t$ (we note that more sophisticated approaches to particle smoothing are available in the literature, but we did not use them in order to make fair comparisons with results from \cite{kucharski2020early} ). Based on these sample we then report in figure \ref{fig:covid}, in the main paper, the mean and credible intervals associated with the following 5 quantities (from \cite{kucharski2020early}). Below $\kappa$ is the rate of reporting, we used the same numerical value as in  \cite{kucharski2020early}, and $F^{(W)}$ and $F^{(T)}$ are auxiliary compartments. 
\begin{enumerate}
	\item The time-varying reproduction number $R_{s}=\tilde{\beta}_{s} \slash \gamma$ for $s =1,\dots,t$;
	
	\item The new confirmed cases by date of onset in Wuhan and China $\tilde{y}^{(3,4)}_{s}$ for $s =1,\dots,t$:
	\begin{equation}
	\hat{z}_{s} \sim \mathrm{Bin} \left (n,  \tilde{p}^{(3,4)}_{s|t} \right ) \quad \text{and} \quad \tilde{y}^{(3,4)}_{s} \sim \mathrm{Bin} \left (\hat{z}_{s},  q^{(W)} \right ).
	\end{equation} 
	\item The new confirmed cases  by date of onset internationally $\tilde{y}^{(7,8)}_{s}$ for $s =1,\dots,t$:
	\begin{equation}
	\hat{z}_{s} \sim \mathrm{Bin} \left (n,  \tilde{p}^{(7,8)}_{s|t} \right ) \quad \text{and} \quad \tilde{y}^{(7,8)}_{s} \sim \mathrm{Bin} \left (\hat{z}_{s},  q^{(T)} \right ).
	\end{equation}
	\item The new confirmed cases by date in Wuhan $\Delta\widetilde{ Conf}^{(W)}_{s}$ for $s =1,\dots,t$:
	\begin{align}
	& \tilde{F}^{(W)}_0 = 0, \\
	& \Delta\tilde{F}^{(W)}_{s} \sim \mathrm{Bin} \left (\tilde{z}_{s}^{(3,4)}, 1-e^{-e^{-\gamma \kappa}} \right ), \quad \Delta\widetilde{ Conf}^{(W)}_{s} \sim \mathrm{Bin} \left (\tilde{F}^{(W)}_{s}, 1-e^{- \kappa} \right ), \\
	& \tilde{F}^{(W)}_{s+1} = \tilde{F}^{(W)}_{s} + \Delta\tilde{F}^{(W)}_{s} - \Delta\widetilde{ Conf}^{(W)}_{s}.
	\end{align}
	\item The new confirmed cases by date internationally $\Delta\widetilde{ Conf}^{(T)}_{s}$ for $s =1,\dots,t$:
	\begin{align}
	& \tilde{F}^{(T)}_0 = 0, \\
	& \Delta\tilde{F}^{(T)}_{s} \sim \mathrm{Bin} \left (\tilde{z}_{s}^{(7,8)}, 1-e^{-e^{-\gamma \kappa}} \right ), \quad \Delta\widetilde{ Conf}^{(T)}_{s} \sim \mathrm{Bin} \left (\tilde{F}^{(T)}_{s}, 1-e^{- \kappa} \right ), \\
	& \tilde{F}^{(T)}_{s+1} = \tilde{F}^{(T)}_{s} + \Delta\tilde{F}^{(T)}_{s} - \Delta\widetilde{ Conf}^{(T)}_{s}.
	\end{align}
	
\end{enumerate}

Figure \ref{fig:covid_kucharski-noflights} shows the results for the SMC algorithm \cite{kucharski2020early} applied to the same data as our method, i.e. with the evacuation flights data left out of the analysis. Therefore figure \ref{fig:covid_kucharski-noflights} can be compared directly against our figure  \ref{fig:covid} - see discussion in the main part of the paper.


\begin{figure}[httb!]
	\centering
	\resizebox{\columnwidth}{!}{
		\includegraphics[scale = 1]{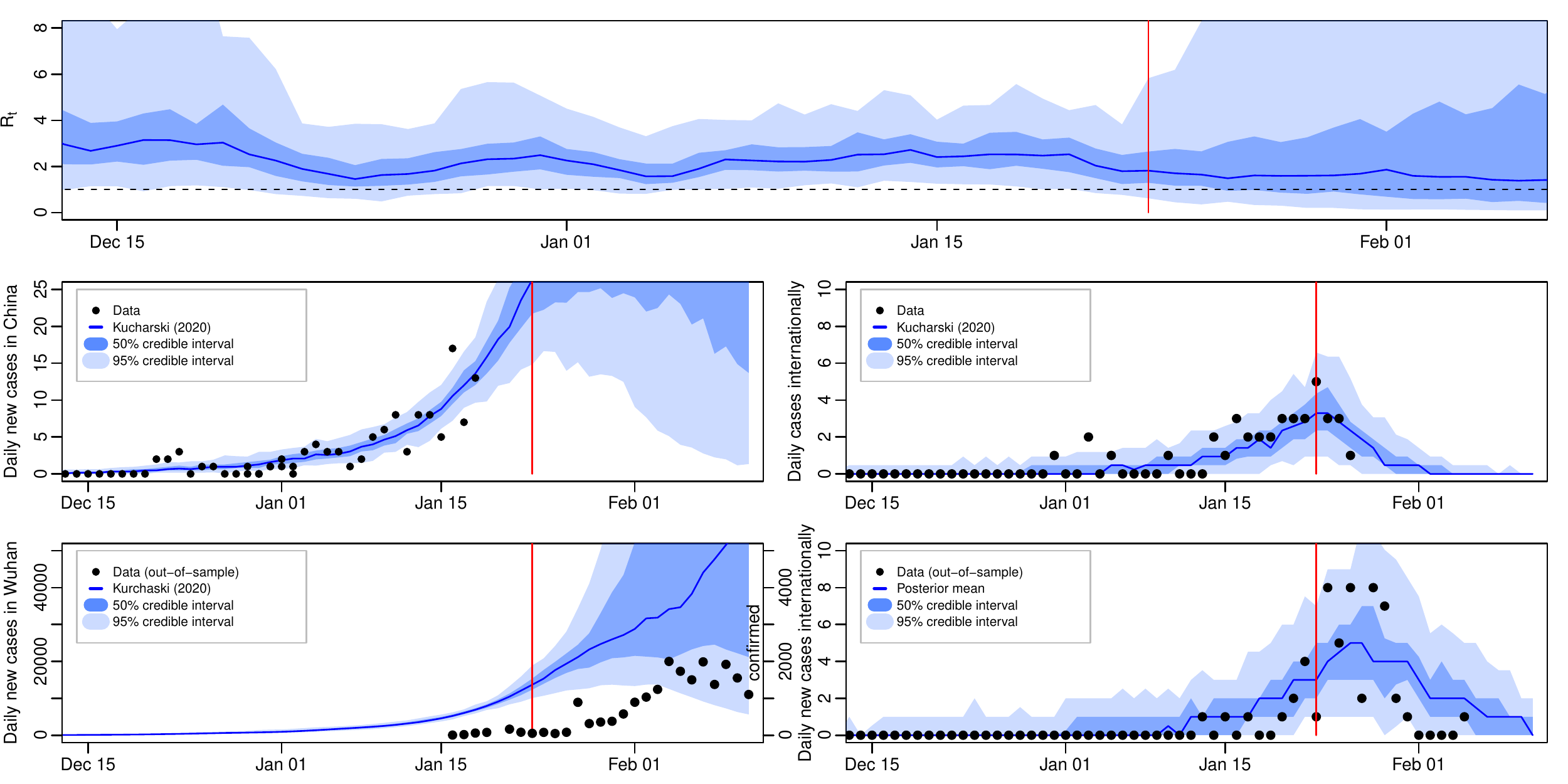}
	}
	\caption{Results for the COVID-19 model using \cite{kucharski2020early} methods without rescue flights data. Red line is date at which travel restrictions were introduced. Top: estimated reproduction number. Middle row:  estimated daily new confirmed cases in Wuhan (left) and internationally (right), both with in-sample data by date of symptom onset. Bottom row, left: estimated new symptomatic but possibly unconfirmed cases (left axis) and out-of-sample new confirmed cases data (right axis); right: estimated confirmed international cases by date of confirmation, and out-of-sample data.} \label{fig:covid_kucharski-noflights}
\end{figure}


\end{document}